\documentclass{llncs}

\usepackage{amsmath}
\usepackage{verbatim}
\usepackage{amssymb}
\usepackage{graphicx}
\usepackage{multirow}
\usepackage{listings}

\numberwithin{equation}{section}
\numberwithin{figure}{section}

\usepackage{enumerate}

\usepackage{etoolbox}

\usepackage{graphicx}

\newcommand{\warn}[1]{\PackageWarning{TomerWarning}{Tom>>>> #1}}

\DeclareMathOperator*{\mybigsqcapDisp}{\mathrel{\reflectbox{\rotatebox[origin=c]{180}{\scalebox{.75}{$\displaystyle{\bigsqcup}$}}}}}

\usepackage{multicol}

\usepackage{wrapfig}

\usepackage{hyperref}

\usepackage{picins}

\usepackage[english]{babel}

\global\long\def\NF{\mathsf{N_{F}}}
\global\long\def\NR{\mathsf{N_{R}}}

\global\long\def\NC{\mathsf{N_{C}}}
\global\long\def\NI{\mathsf{N_{n}}}

\global\long\def\mm{\mathcal{M}}
\global\long\def\pp{\mathcal{P}}

\global\long\def\op#1#2{{#1}_{{#2}\rhd}}
\global\long\def\nn{\mathcal{N}}

\global\long\def\ovtpn{\overline{tp}}

\global\long\def\LoiqNObTEXT{ALCQIO}
\global\long\def\ctwotree{CT^{2}}
\global\long\def\ctwo{C^{2}}

\global\long\def\L{\mathcal{\LoiqNObTEXT}_{b,Re}}
\global\long\def\LoiqNOb{\mathcal{\LoiqNObTEXT}}
\global\long\def\Loiq{\LoiqNOb_{b}}
\global\long\def\LoiqTEXT{\LoiqNObTEXT b}
\global\long\def\val{\mathit{val}}
\global\long\def\ORDW{\mathrm{ORD}}
\global\long\def\TYPES{\mathrm{TYPES}}
\global\long\def\Reacha#1#2#3{#1\scalebox{.5}{\ensuremath{\stackrel{\scalebox{1.5}{\ensuremath{\boldsymbol{\longrightarrow}}}}{\boldsymbol{\subseteq}}}}^{#2}#3}

\newtheorem{thm-th:exp}{Theorem \ref{th:exp}}

\makeatletter
\def\@copyrightspace{\relax}
\makeatother

\begin{document}

\setlength{\pdfpageheight}{\paperheight}
\setlength{\pdfpagewidth}{\paperwidth}





\title{Extending \texorpdfstring{$\LoiqNOb$}{\LoiqNObTEXT} with reachability}
\subtitle{A description logic for shape analysis}

\author{Tomer Kotek \and Mantas \v{S}imkus \and Helmut Veith \and Florian
Zuleger}

\institute{Vienna University of Technology \\\email{ \{kotek,veith,zuleger\}@forsyte.at,
simkus@dbai.tuwien.ac.at}}

\maketitle

\begin{abstract}

We introduce a description logic $\L$ which adds reachability assertions to $\LoiqNOb$, a sub-logic of the two-variable fragment of first order logic with counting quantifiers.
$\L$ is well-suited for applications in software verification and shape analysis.
Shape analysis requires expressive logics which can express reachability and have good computational properties.
We show that $\L$ can describe complex data structures with a high degree of sharing and allows compositions such as list of trees.

We show that finite satisfiability and finite implication of $\L$- formulae are polynomial-time reducible to finite satisfiability of $\LoiqNOb$-formulae.
As a consequence, we get that finite satisfiability and finite implication in $\L$ are NEXPTIME-complete.
Description logics with transitive closure constructors have been studied before, but $\L$ is the first description logic that remains decidable on finite structures while allowing at the same time nominals, inverse roles, counting quantifiers and reachability assertions,
\end{abstract}




\warn{thanks}
\warn{spellcheck!}

\newcommand{\useful}{\ensuremath{\mathit{ext}}}
\newcommand{\order}{\ensuremath{\mathit{ord}}}
\newcommand{\nominals}{\ensuremath{O}}
\newcommand{\usefulExp}{\ensuremath{\mathit{ext}}}
\newcommand{\orderExp}{\ensuremath{\mathit{ord}}}
\newcommand{\nominalsExp}{\ensuremath{O}}
\newcommand{\st}{\ensuremath{\mathit{start}}}
\newcommand{\suc}{\ensuremath{\mathit{succ}}}
\newcommand{\ORDWExp}{\ensuremath{\mathit{ORD}}}
\newcommand{\eval}{\ensuremath{\mathit{eval}}}
\newcommand{\formula}[1]{\ensuremath{\mathit{assoc}(#1)}}
\newcommand{\semi}[1]{\ensuremath{\mathit{semi}(#1)}}
\newcommand{\reachCyclic}{\ensuremath{\mathit{reach-cyc}}}
\newcommand{\reach}[2]{\ensuremath{\mathit{Reach}_{#1}(#2)}}
\newcommand{\marked}{\ensuremath{\mathit{Marked}}}
\newcommand{\ghost}[1]{\ensuremath{#1^\mathit{ghost}}}

\section{Introduction}
\emph{Description Logics} (\emph{DLs}) are a well established family of logics for Knowledge Representation and Reasoning \cite{bk:Baader2003}.
They model the domain of interest in terms of \emph{concepts} (classes of objects) and \emph{roles} (binary relations between objects).
These features make DLs very useful to formally describe and reason about graph-structured information.
The usefulness of DLs is witnessed e.g.\,by the W3C choosing DLs to provide the logical
foundations to the standard Web Ontology Language (OWL)~\cite{owl2-overview}.
Another application of DLs is formalization and static analysis of UML class diagrams and ER diagrams, which are basic modeling artifacts in object-oriented software development and database design, respectively~\cite{BeCD05,ACKRZ07b}.
In these settings, standard reasoning services provided by DLs can be used to verify e.g.\, the consistency of a diagram.

In complex software projects, the source code is usually accompanied by design documents which provide extensive documentation and models of data structure content (e.g., using UML and ER diagrams).
This documentation is both an opportunity and a challenge for program verification.
However, the verification community has focused mostly on a bottom-up approach to the analysis of programs with dynamic data structures, which examines pointers and the shapes induced by them.
It is therefore a compelling question, if DLs can be used to model and verify richer properties of dynamic data structures.

A promising framework for shape analysis based on description logics was developed in \cite{georgieva2005description}.
\cite{georgieva2005description} is based mostly on
an extension of $\LoiqNOb$ with fixed points.
$\LoiqNOb$ is itself the extension of $\mathcal{ALC}$ with nominals, number restrictions and inverses, see e.g. \cite{DBLP:conf/aaai/BaaderLMSW05}.
$\mathcal{ALC}$ is a syntactic variant of the multimodal logic $K_m$~\cite{DBLP:journals/ai/Schmidt-SchaussS91,Schild91acorrespondence}.
$\LoiqNOb$ already fulfills most of the requirements for describing the memory of programs with dynamic data structures:
(i) nominals allow to represent the program's variables;
(ii) number restrictions allow to model the program's pointers by functions;
(iii) inverses allow to define the incoming pointers for data structure elements, e.g., in a tree every tree element must have at most one parent.
Additionally, it needs to be ensured that data structures only contain elements \emph{reachable} from program variables via program pointers (iv).
Reachability is expressible in $\mu \LoiqNOb$, the extension of $\LoiqNOb$ with fixed points, but not in $\LoiqNOb$.

The main disadvantage of \cite{georgieva2005description} is that finite satisfiability and implication of $\mu \LoiqNOb$ formulas
is undecidable \cite{ar:BonattiPeron2004}.
No extensions of $\LoiqNOb$ with reachability or transitive closure
were known to be decidable on finite structures.

{\em Our contribution: }
We introduce
and develop decision procedures over finite structures
for the logic $\L$, which extends the closure ($\Loiq$) of $\LoiqNOb$
under Boolean operations with reachability assertions.
The reachability assertions guarantee that elements of the universe of a model are reachable in the graph-theoretic sense from initial sets of elements using prescribed sets of binary relation symbols.
Alternatively, we can think of $\L$ as $\Loiq$ interpreted over structures containing an unbounded number of binary trees.

The main results of this paper are algorithms which decide the finite satisfiability and finite implication problems of $\L$.
The algorithms are reductions to finite satisfiability in $\LoiqNOb$, which suggests relatively simple implementation using existing $\LoiqNOb$ reasoners.
The algorithms run in NEXPTIME, which is optimal since $\LoiqNOb$ is already NEXPTIME-hard.

{\em A description logic for shape analysis.}
The logic $\L$ we introduce is especially well-suited to shape analysis, since $\L$ contains
nominals, number restrictions, inverses and reachability.
$\L$ is a flexible and powerful formalism for describing complex data structures with sharing.
This, together with the existence of a decision procedure for implication and not just satisfiability, makes $\L$ a promising candidate for software verification applications.

We discuss in Section \ref{se:examples} that $\L$ is strong enough to describe e.g. lists, trees and lists of trees, etc.
$\L$ supports programs with sharing, in which memory cells (which in model-theoretic terms are elements of the universe of the model) may participate in multiple data structures.
In particular, we show that $\L$ supports modular reasoning:
Our results support composition of data structures with disjoint domains, as well as composition of data structures whose domains are not disjoint, but have disjoint pointers.

The closure of the underlying logic $\Loiq$ under Boolean operations allows to describe conditional statements in programs.
The decision procedure for implication for $\L$ is essential for verification applications, since it allows to show that specifications relating pre- and post-conditions are correct.
We demonstrate the usefulness of $\L$ by giving a Hoare-style correctness proof for an example pointer-manipulating program in Section~\ref{se:verification}.
We leave the development of a full Hoare-style verification framework for future work.

\paragraph{Related Work.}
We discuss further related work in the following.

\emph{Separation Logic.}
In the last decade, the leading approach for shape analysis has been separation logic~\cite{pr:Reynolds02}.
While separation logic has a strong proof-theoretic tradition - in origin and in proof techniques used, model-theoretic approaches have been less prominent.
We believe that recent advances in finite model theory have created an opportunity for the development of new model-theoretic approaches to shape analysis.
In this paper, we explore this model-theoretic line of research building on previous results on description logics.

\emph{Description logics and reachability. }
Description logics extended with various forms of reachability have been studied in the literature, though the focus is mostly on arbitrary
rather than finite structures.
The important work of Schild~\cite{Schild91acorrespondence} exposed a correspondence between variants of propositional dynamic logic (PDL), a logic for reasoning about program behavior, and variants of DLs extended with further role constructors, e.g.\,the transitive closure of a role.
Close correspondences between DLs extended with fixpoints and variants of the $\mu$-calculus have also been identified~\cite{DBLP:conf/kr/Schild94,DBLP:conf/ecai/GiacomoL94}.
Recently, extensions of DLs with regular expressions over roles have been proposed~\cite{calv-etal-09}.
DLs with transitive roles and counting quantifiers were studied in  \cite{KSZ07,schroder2008many}.
\cite{pr:DLC13} proved decidability over arbitrary structures for a DL with transitive closure and counting quantifiers.

{\em The two-variable fragment.}
Our results bear similarity with a recent deep result~\cite{pr:CharatonikWitkowski13} based on \cite{ar:PrattHartmann2005}.
There, the complexity of finite satisfiability of the two-variable fragment of first order logic extended with counting quantifiers ($\ctwo$) and additionally with two forests ($\ctwotree$) is studied.

The results in our paper and in \cite{pr:CharatonikWitkowski13} are incomparable due to differences in several orthogonal aspects.
(i) $\ctwo$ strictly contains $\Loiq$.
(ii) $\ctwotree$ is restricted to at most two forests, whereas $\L$ allows an unbounded number of reachability conditions.
The decidability of the extension of $\ctwo$ with three successor relations is not known, while extending $\Loiq$ with three successor relations is covered by our results.
(iii) we have a decision procedure for implication in $\L$, while no such decision procedure is given for $\ctwotree$ in \cite{pr:CharatonikWitkowski13}.
(iv) to our knowledge, no reasoners for $\ctwo$ exist; the sophisticated construction in \cite{pr:CharatonikWitkowski13} makes the worthy task of implementing a reasoner for $\ctwotree$ a considerable challenge.
In contrast, our result reduces reasoning in $\L$ to satisfiability in $\LoiqNOb$, which is contained in the description logic $\mathcal{SROIQ}$ for which several reasoners have been implemented, e.g., \cite{pellet,fact,hermit}.

Due to the intricate nature of the proof in \cite{pr:CharatonikWitkowski13}, the exact relationship between our result and that of \cite{pr:CharatonikWitkowski13} is difficult to ascertain.
It would be beneficial in future work to understand whether these results can be united within a natural logic containing both $\ctwotree$ and $\L$.

\emph{Shape Analysis and Content Verification.}
In a recent paper~\cite{arXiv:CKSVZ13} the use of a DL to reason and verify correctness of entity-relations-type content of data structures on top of an existing shape analysis is discussed.
Technically, \cite{arXiv:CKSVZ13} is based on a reduction of a DL to satisfiability in $\ctwotree$.
The DL in \cite{arXiv:CKSVZ13} cannot express reachability and the approach there depends on a combination of DL with an existing shape analysis.
We believe the method of \cite{arXiv:CKSVZ13} can be modified to be based on $\L$ alone.
Exploration of this question is part of future work.

\section{The Formalism and Examples}

From the point of view of finite model theory, $\Loiq$ and $\L$
are syntactic variants of fragments of first or second order logic.
In description logics terminology, binary relation symbols are called
{\em atomic roles}, unary relation symbols are called {\em atomic concepts},
and constant symbols are called {\em nominals}. Let $\NR$, $\NC$
and $\NI$ denote the sets of atomic roles, atomic concepts and nominals.
A vocabulary $\tau$ is then the union of $\NR$, $\NC$ and $\NI$.
Let $\NF\subseteq\NR$ be a set of atomic roles. The roles in $\NF$
are called {\em functional}.

Formulae are built from the symbols in $\tau$. The various constructors
available to build formulae determine the particular DL,
giving rise to a wide family of logics with varying expressivity,
and decidability and complexity of reasoning. The semantics to formulae
is given in terms of structures, where atomic concepts and atomic
roles are interpreted as unary and binary relations in a structure,
respectively, and constants are interpreted as elements in the structure's
universe.
We now define $\Loiq$ and $\L$.
See Section \ref{se:examples} for examples.
\begin{definition}
[Syntax of $\Loiq$] The set of \emph{roles},\emph{ concepts} and\emph{
formulae} of $\Loiq$ are defined inductively:
\begin{itemize}
\item Atomic concepts and nominals are concepts; Atomic roles are roles;
\item If $r$ is a role, $C,D$ are concepts and $n$ is a positive integer,
then $C\sqcap D$, $C\sqcup D$, $\neg C$, $\exists r.C$ and $\exists^{\leq n}\, r.C$
are concepts, and $r^{-}$ is a role;
\item $C\sqsubseteq D$ (\emph{concept inclusion}) and $C\equiv D$ (\emph{concept equality}) where $C,D$ are concepts,
are formulae;
\item If $\varphi$ and $\psi$ are formulae, then $\varphi\land\psi$,
$\varphi\lor\psi$, and $\neg\psi$ are formulae.
\end{itemize}
\end{definition}
The sub-logic $\LoiqNOb\subseteq \Loiq$ does not allow negations and disjunctions.

A \emph{structure} (or \emph{interpretation}) is a tuple $\mm=(M,\tau,\cdot)$,
where $M$ is a finite set (the \emph{universe}), $\tau$ is a set
of \emph{constants} and unary and binary \emph{relation symbols} (\emph{the
vocabulary}), and $\cdot$ is an \emph{interpretation function}, which
assigns to each constant $c\in\tau$ an element $c^{\mm}\in M$, and
to each $n$-ary relation symbol $R\in\tau$ an $n$-ary relation
$R^{\mm}$ over $M$. In this paper, each relation is either unary
or binary (i.e. $n\in\{1,2\}$). In this paper, all structures are
finite. Satisfiability and implication always refer to finite structures
only.
\begin{definition}
[Semantics of $\Loiq$ ]
The semantics of an $\Loiq(\tau)$-formula $\varphi$ is given in terms of $\tau$-structures such that every $f\in\NF$, $f^{\mm}$ is a partial function.
The function $\cdot^{\mm}$ is extended to the remaining concepts and roles inductively below.
The satisfaction relation $\models$ is also given below.
If $\mm\models\varphi$, then $\mm$ is a \emph{model} of $\varphi$.
We write $\psi\models\varphi$ and say that $\psi$ implies $\varphi$ if every model of $\psi$ is also a model of $\varphi$.
\[
\begin{array}{llllll}
(C\sqcap D)^{\mm} & = & C^{\mm}\cap D^{\mm} \\
(C\sqcup D)^{\mm} & = & C^{\mm}\cup D^{\mm} \\
(\neg C)^{\mm} & = & M\setminus C^{\mm} \\

(r^{-})^{\mm} & = & \{(e,e')\mid (e',e)\in r^{\mm})\}\\
(\exists r.C)^{\mm} & = & \{e\mid\exists e':(e,e')\in r^{\mm},\,e'\in C^\mm\}\\
(\exists^{\leq n}r.C)^{\mm} & = & \{e\mid\exists^{\leq n}e':(e,e')\in r^{\mm},\,e'\in C^\mm\}\\

\mm\models C\sqsubseteq D & \mathrm{if} & C^{\mm}\subseteq D^{\mm} \\
\mm\models\varphi\land\psi & \mathrm{if} & \mm\models\varphi~\mathrm{and}~\mm\models\psi\\
\mm\models\neg\varphi & \mathrm{if} & \mm\not\models\varphi\\
\mm\models\varphi\lor\psi & \mathrm{if} & \mm\models\varphi~\mathrm{or}~\mm\models\psi
\end{array}
\]

\end{definition}
We will use the following abbreviations.
$\top=C\sqcup \neg C$, where $C$ is an arbitrary atomic concept and $\bot=\neg \top$;
$\alpha \equiv \beta$ for the formula $\alpha\sqsubseteq \beta\land \beta\sqsubseteq \alpha$;
$\exists r$ for the concept $\exists r.\top$;
$\exists^{= n}r.C$ for the concept $\exists^{\leq n}r.C \sqcap \neg \exists^{\leq n-1}r. C$;
Note that $\top^\mm=M$ and $\bot^\mm=\emptyset$ for any structure $\mm$ with universe $M$.
For a formula $\varphi$, we denote by $|\varphi|$ the length of $\varphi$ as a string.

For $\L$, we define two new types of assertions.
\begin{description}
\item [{Reachability~Assertion}] $\Reacha{B}{S}{A}$ where $A,B\in\NC$
and $S\subseteq\NF$.
Intuitively, it says that $B$ is contained in $A$ and that $A$ is a set of elements reachable from $B$, without leaving $A$, through the roles of $S$.
\item [{Disjointness~Assertion}] $Disj(A_{1},A_{2})=(A_{1}\sqcap A_{2}\equiv\bot)$
for $A_{1},A_{2}\in\NC$\textit{.}
\end{description}
Let $RE$ and $DI$ be sets of reachability respectively disjointness
assertions.
\begin{description}
\item [{Compatibility}] $RE$ and $DI$ \emph{compatible }if for every
$\Reacha{B_{1}}{S_{1}}{A_{1}}$ and $\Reacha{B_{2}}{S_{2}}{A_{2}}$
in $RE$ such that $S_{1}\cap S_{2}\not=\emptyset$, $Disj(A_{1},A_{2})$
is in $DI$.
\end{description}

\begin{definition}
[Syntax of $\L$]  $\Phi=\phi\land\bigwedge RE\land\bigwedge DI$ is an $\L$-formula if

 \begin{enumerate}[(A)]
\item $\phi \in \Loiq$
\item $RE$ is a set of reachability assertions
\item $DI$ is a set of disjointness assertions
\item $RE$ and $DI$ are compatible
\end{enumerate}
The set of {\em containment assertions}
\[CO(RE)=\left\{ B\sqsubseteq A\mid \Reacha{B}{S}{A}\in RE\right\}\,.\]
\end{definition}

\begin{definition}
[Semantics of $\L$]
Let $\Phi=\phi\land\bigwedge RE\land\bigwedge DI\in\L$.
Let $\formula{\Phi} = \phi \land \bigwedge CO(RE) \land  DI$ be the $\Loiq$ formula associated to $\Phi$.

For every $\beta=\Reacha{B}{S}{A}\in RE$ and $\tau$-structure $\mm$,	
let $D_{\beta}^{\mm}$ be the directed subgraph $\left\langle M,\bigcup_{s\in S}s^{\mm}\right\rangle$ induced by $A^{\mm}$; we call $D_{\beta}^{\mm}$ \emph{connected}, if every vertex of $D_{\beta}^{\mm}$ is reachable from a vertex of $B^{\mm}$.

For every $\tau$-structure $\mm$, $\mm\models\Phi$ if $\mm \models \formula{\Phi}$ (in $\Loiq$ semantics) and $D_{\beta}^{\mm}$ is connected for every $\beta = \Reacha{B}{S}{A}\in RE$.
\end{definition}

\paragraph{Comment 1:}
For illustration purposes, we state an equivalent definition of the semantics of $\L$ in the following:
Let $*$ denote the reflexive-transitive closure, $\circ$ denote role composition,
$R^\mm = (\bigcup_{s\in S}s^{\mm}\cap A^{\mm}\times A^{\mm})^*$ and
$R^\mm(B^\mm) = \{ v \mid \exists (u,v)\in R^\mm.\, u\in B^\mm\}$.
We have $\mm\models\Phi$ iff $\mm \models \formula{\Phi}$ and $R^\mm(B^\mm) = A^M$ for every $\beta = \Reacha{B}{S}{A}\in RE$.

\paragraph{Comment 2:}
The motivation for $CO(RE)$ is to guarantee that data structures over the same roles
have disjoint domains (see {\em Compositionality 1} in the examples).

%

\subsection{Examples} \label{se:examples}

\paragraph{List-segments and successor relations.}
Given a concept $L$, a nominal $\mathit{head}$ and a functional role $\mathit{next}$,
$L^{\mm}$ is a singly-linked list-segment from $\mathit{head}^{\mm}$ via $\mathit{next}^{\mm}$,
if the directed subgraph of $\left\langle M,\mathit{next}^{\mm}\right\rangle$ induced by $L^{\mm}$ is a (potentially cyclic) successor relation.
This can be expressed by the $\L$ formula $\Phi^{\mathit{head},\mathit{next},L}_{List}$ obtained as the conjunction of
$\top\sqsubseteq \top$,
$RE_{l}=\{\Reacha{head}{next}{L}\}$ and $DI_{l}=\emptyset$.
$RE_l$ expresses that $\mathit{head}^\mm \in L^\mm$ and all elements of $L^{\mm}$
are reachable from $\mathit{head}^{\mm}$ via $\mathit{next}^{\mm}$; $DI_l$ is empty
since we have described no other data structure which could be disjoint
from this list.
$\Phi^{\mathit{head},\mathit{next},L}_{List}$ does not determine where the $\mathit{next}$ role of the last
element of the list points. Acyclic and cyclic list-segments are defined as follows:
$\Phi^{\mathit{head},\mathit{next},L}_{aList} = \Phi^{\mathit{head},\mathit{next},L}_{List} \land \neg(L\sqsubseteq \exists next)$
and
$\Phi^{\mathit{head},\mathit{next},L}_{cList} = \Phi^{\mathit{head},\mathit{next},L}_{List} \land
(\mathit{head}\sqsubseteq \exists \mathit{next}^-.L)$.

\paragraph{$d$-ary trees.}
Given a concept $T$, a nominal $\mathit{root}$ and functional roles $\mathit{left}$
and $\mathit{right}$, $T^{\mm}$ is a binary tree rooted at $root^{\mm}$
via $\mathit{left}^{\mm}$ and $\mathit{right}^{\mm}$ if the directed subgraph of
$\left\langle M,\mathit{left}^{\mm}\cup \mathit{right}^{\mm}\right\rangle $
induced by $T^{\mm}$ is a directed tree rooted at $\mathit{root}^{\mm}$ in
the graph-theoretic sense. This can be expressed by the $\L$ formula
$\varPhi_{T}$ obtained as the conjunction of
(a) $\mathit{root} \sqsubseteq \neg\exists \mathit{left}^-.T \sqcap \neg\exists \mathit{right}^-.T$,
(b) $T \sqcap \neg \mathit{root} \sqsubseteq \exists^{= 1}\mathit{left}^{-}.T \sqcap \neg\exists \mathit{right}^{-}.T \sqcup \exists^{= 1}\mathit{right}^{-}.T \sqcap \neg\exists \mathit{left}^{-}.T$,
(c) $RE_{t}=\{\Reacha{root}{\{\mathit{left},right\}}{T}\}$ and (d) $DI_{t}=\emptyset$.
(a) expresses that $\mathit{root}^{\mm}$ belongs to $T^{\mm}$ and is not
pointed to from $T^{\mm}$; (b) expresses that every element of $T^{\mm}$
besides the root has exactly one incoming pointer from a $T^{\mm}$
element. (c) expresses that $\mathit{root}^{\mm}$ belongs to $T^{\mm}$ and that all elements of $T^{\mm}$ are reachable from $\mathit{root}^{\mm}$ via $\mathit{left}^{\mm}$ and $\mathit{right}^{\mm}$.
The case of $d$-ary trees, $d>2$, is similar, using
$d$ functional roles $child_1,\ldots,child_d$.
We remark that $1$-ary trees correspond to acyclic list-segments.

\paragraph{Compositionality 1.}
$\L$ is closed under taking memory disjoint union of data structures.
E.g. $\L$-formulae $\Phi_i=\varphi_i\land \{\Reacha{B_i}{S_i}{A_i}\}$, $i=1,2$,
the following formula expresses that the domains of its models consist of
two disjoint parts, corresponding to the $\Phi_i$:
$\Phi = \varphi_1\land\varphi_2\land \bigwedge RE\land \bigwedge	 DI$,
where $RE = \{\Reacha{B_i}{S_i}{A_i}\mid i=1,2\}$ and
$DI = \{Disj(A_{1},A_{2})\}$.
There is no disjointness requirement on the roles in $S_1$ and $S_2$.

\paragraph{Compositionality 2.}
$\L$ allows to define multiple data structures which may overlap in memory,
as long as they do not share the same pointers.
E.g. given $\L$-formulae $\Phi_i=\varphi_i\land \{\Reacha{B_i}{S_i}{A_i}\}$, $i=1,2,3$,
such that $S_1$, $S_2$ and $S_3$ are pairwise disjoint,
the following formula expresses that the three data structures are defined simultaneously
with possibly overlapping domain:
$\Phi = \varphi_1\land\varphi_2\land \phi_3\land \bigwedge RE$,
where $RE = \{\Reacha{B_i}{S_i}{A_i} \mid i=1,2,3\}$.
We note that this compositionality property allows us to define three successor relations.

\paragraph{Compositional Data Structures.}
$\L$ allows us to define compositional data structures such as list of lists, list of trees, tree of lists of lists, etc.
For example, let us define an acyclic list of acyclic lists:
Given a concept $L$, a nominal $head$, and functional roles $next_1$
and $next_{2}$, $L^{\mm}$ is a acyclic list of acyclic lists from $head^{\mm}$
via $next_1^{\mm}$ and $next_{2}^{\mm}$ if there exists $L_1^\mm \subseteq L^{\mm}$
such that $L_{1}^{\mm}$ is an acyclic list from $var_{1}^{\mm}$ via $next_1^{\mm}$,
and $L^{\mm}$
is a disjoint union of acyclic lists via $next_{2}^{\mm}$ whose heads belong to $L_{1}^{\mm}$.
This can be expressed by the acyclic list-segment formulae $\Phi^{\mathit{head},\mathit{next_1},L_1}_{aList}$ and $\Phi^{L_1,\mathit{next_2},L}_{aList}$, which can be composed because of disjoint roles in the reachability assertions (see Compositionality 2).


\subsection{Verification of Pointer-manipulating Programs using $\L$ }
\label{se:verification}

The focus of this paper is the development of a logic and decision procedures which can be used for shape analysis in future work.
However, we believe that it is important to relate the logic we develop in this paper to its intended application.
We illustrate in this section how to use $\L$ for the verification of pointer-manipulating programs.
We will discuss on the example program given in Fig. \ref{fig:prog} how to formulate verification conditions for a Hoare-style correctness proof.
We leave the non-trivial task of developming a full verification framework based on $\L$ for future work.
Section \ref{se:main-section} can be read independently of this section.

The program in Fig. \ref{fig:prog} receives as input a list pointed to by $head$ in which some elements may be marked but $head$ is not marked.
The program removes from the list all marked elements.

\begin{figure}
\begin{lstlisting}
start: b := head;
loop:  while (b != null)
          if (b.next != null && b.next.marked)
             b.next := b.next.next;
          else
             b := b.next;
end:   ;
\end{lstlisting}

\caption{\label{fig:prog}}
\end{figure}

We annotate the labels $\mathsf{start}$, $\mathsf{end}$ and $\mathsf{loop}$ with the pre-condition, the post condition and the loop invariant respectively:
\begin{description}
\item [{[Pre-condition]}] $head$ points to the first element of an acyclic list.
    $head$ is not marked.
\item [{[Post-condition]}]~
    \begin{enumerate}
        \item After the program is executed, $head$ points to the first element of an acyclic list containing only unmarked elements.
        \item Moreover, the elements removed from the input list are exactly the marked elements.
    \end{enumerate}
\end{description}
The following loop invariant suffices to prove the correctness of the program with respect to the pre- and post-conditions:
\begin{description}
\item [{[{Loop invariant}]}]~
    \begin{enumerate}
        \item The pre-condition holds.
        \item $b$ points to an element of the list.
        \item the set of unmarked elements in the list is exactly the same set as at the beginning of the program.
        \item all elements in the list before $b$ are unmarked.
    \end{enumerate}
\end{description}
The correctness of the program is achieved by proving that the annotations are correct.
More precisely, $S_{2}$ is the piece of loopless code inside the while loop, i.e., the sequence of $\mathsf{assume(b  \ {!}{=} \ null)}$ and the code in the if-then-else statement, $S_{1}$ is the assignment $\mathsf{b := head}$ and $S_{3}$ is the statement $\mathsf{assume(b \ {=}{=} \  null)}$.
To prove the correctness of the program we need to prove that:
\begin{description}
\item [{[$\mathsf{VC}_1$]}] If the pre-condition holds, then after executing $S_{1}$ the loop invariant will hold.
\item [{[$\mathsf{VC}_2$]}] If the loop invariant holds, then after executing on iteration of the loop, i.e., executing $S_{2}$ once, the loop invariant will hold again.
\item [{[$\mathsf{VC}_3$]}] If the loop-invariant holds, then after executing $S_{3}$ the post-condition will hold.
\end{description}

\paragraph{Annotations in $\L$: }

Next we write the annotations in $\L$. The memory is represented
as $\rho$-structures and $\tau$-structures with $\rho\subseteq\tau$.
$\rho$ contains the nominals $head$, $b$ and $Null$, the role
$next$ and the concept $\marked$. We think of each element of the
memory as a fixed size block of memory containing the $\mathsf{next}$
pointer field and the $\mathsf{marked}$ field, except for the special
element which interprets $Null$ and represent the value $\mathsf{null}$.
$head$ and $b$ are the elements which $\mathsf{head}$ and $\mathsf{b}$
point to. $next$ is a function on the elements of the memory defined
by the pointer $\mathsf{next}$. $\marked$ contains the set of elements
whose $\mathsf{marked}$ field is set to true.

Our presentation here is a simplification of the memory model in \cite{arXiv:CKSVZ13},
which further supports dynamic allocation and deallocation memory.

$\tau$ extend $\rho$ with the concepts $L$ 
and $L_{b}$.
$L$ will contain the elements of the list.
$L_{b}$ will contain
the elements of the list segment starting at $head$ up to $b$ (not
including $b$).

The annotations $\Phi_{pre}$, $\Phi_{l-inv}$ and $\Phi_{post}$ refer to the memory at the labels $\mathsf{start}$, $\mathsf{loop}$ and $\mathsf{post}$ respectively.
The pre-condition $\Phi_{pre}$ is
\[
\Phi_{pre} = \Phi_{aList}^{head,next,L} \land head \sqsubseteq \neg \marked
\]

$\Phi_{aList}^{head,next,L} = \Reacha{head}{next}L\land\neg(L\sqsubseteq\exists next)$ is the formula defining an acyclic list from Section \ref{se:examples}.
The post-condition is given by $\Phi_{post}=\Phi_{post1}\land\Phi_{post2}$, where
\begin{eqnarray*}
\Phi_{post1} & = & \Phi_{aList}^{head,next,L}\land L\sqsubseteq\neg \marked\\
\Phi_{post2} & = & L\equiv L_{ghost}\sqcap\neg \marked_{ghost}
\end{eqnarray*}
$L_{ghost}$ and $\marked_{ghost}$ represent the values of $L$ and $\marked$ at the start of the program.
\[
\Phi_{l-inv}=\Phi_{l-inv1}\land\Phi_{l-inv2}\land\Phi_{l-inv3}\land\Phi_{l-inv4}
\]
where
\begin{eqnarray*}
\Phi_{l-inv1} & = &  \Phi_{pre}\\
\Phi_{l-inv2} & = & b\sqsubseteq L\\
\Phi_{l-inv3} & = & L \sqcap \neg \marked \equiv L_{ghost} \sqcap \neg \marked_{ghost}\\
\Phi_{l-inv4} & = & (L_{b}\sqsubseteq L) \land (head\sqsubseteq L_{b}\sqcup b)\land(b\sqsubseteq\neg L_{b})\\
 &  & \land(L_{b}\sqsubseteq\forall next.L_{b}\sqcup b)\land(\neg L_{b}\sqsubseteq\forall next.\neg L_{b})\\
 &  & \land(L_{b}\sqsubseteq\neg \marked)
\end{eqnarray*}

$\Phi_{l-inv4}$ expresses that $L_{b}$ is exactly the set of elements
in $L$ from $head$ to $b$, not including $b$.

\paragraph{Verification conditions in $\L$: }

Expressing the verification conditions exactly requires us first to
relate the loopless pieces of code $S\in\{S_{1},S_{2},S_{3}\}$ with
the annotations $\Phi_{pre}$, $\Phi_{post}$ and $\Phi_{l-inv}$.
Each $S_{i}$ is associated with two annotations $\Phi_{start}^{i},\mbox{\ensuremath{\Phi}}_{end}^{i}\in\{\Phi_{pre},\Phi_{post},\Phi_{l-inv}\}$
(e.g., for $S_{1}$ we have $\Phi_{start}^{1}=\Phi_{pre}$ and $\Phi_{end}^{1}=\Phi_{l-inv}$).

We need three formulas $\Psi_{1},\Psi_{2},\Psi_{3}\in\L(\tau)$ such
that $\psi_{i}$ is satisfiable iff $\mathsf{VC}_{i}$ does not hold.
The verification conditions $\mathsf{VC}_{1}$, $\mathsf{VC}_{2}$,
and $\mathsf{VC}_{3}$ refer to the memory both at the start and at
the end of $S_{i}$. In contrast, the truth-values of $\Psi_{1}$,
$\Psi_{2}$ and $\Psi_{3}$ will be evaluated on $\tau$-structures
corresponding to the memory at the start of $S_{i}$ only.

The main observation which is required to write $\Psi_{i}$ is that
the symbols of $\tau$ in $\Phi_{end}^{i}$ referring to the memory
at the end of $S_{i}$ can be written in terms of the same sybmols
when they refer to the memory at the start of $S_{i}$.

Let $\underline{\rho}$ and $\underline{\tau}$ be disjoint copies
of $\rho$ and $\tau$ with symbols $\underline{sym}$ corresponding
to $sym$. Let $\mm$ be a $\rho$-structure. Let $\mm_{S_{i}}$ be
the unique $\underline{\rho}$-structure obtained by executing $S_{i}$
on $\mm$ and renaming all symbols $sym$ to $\underline{sym}$. Similarly,
we write $\nn_{S_{i}}$ for the unique $(\tau\backslash\rho)\cup\underline{\rho}$-structure
obtained by replacing the sub-structure $\mm$ with vocabulary $\rho$
of $\nn$ with $\mm_{S_{i}}$.
Let $\underline{\Phi_{end}^{i}}$ be
obtained from $\Phi_{end}^{i}$ be replacing each $\tau$-symbol $sym$
with the $\underline{\tau}$-symbol $\underline{sym}$.

We will define $\Psi_{i}$ by expressing all symbols of $\underline{\tau}$
referring to the memory at the end of $S_{i}$ by $\tau$-symbols
referring to the memory at the start of $S_{i}$.

Consider $\mathsf{VC}_{i}$. The following are equivalent:
\begin{itemize}
\item [(A)] $\mathsf{VC}_{i}$ does not hold
\item [(B)] There exist a $\tau$-structure $\nn$ such that $\nn\models\Phi_{start}^{i}$
and for every extension $\pp'$ of the $\tau\cup\underline{\rho}$-structure
$\pp=\left\langle \nn,\nn_{S_{i}}\right\rangle $ to a $\tau\cup\underline{\tau}$-structure
we have $\pp'\not\models\underline{\Phi_{end}^{i}}$.
\end{itemize}
Let $\phi_i = \formula{\Phi_{end}^{i}}$. 
 (A) and (B) are equivalent
to (C):
\begin{itemize}
\item [(C)] There exists a $\tau$-structure $\nn$ such that $\nn\models\Phi_{start}^{i}$
and for every extension $\pp'$ of the $\tau\cup\underline{\rho}$-structure
$\pp=\left\langle \nn,\nn_{S_{i}}\right\rangle $ to a $\tau\cup\underline{\tau}$-structure
such that $\pp'\models\bigwedge RE\land\bigwedge DI$, we have $\pp'\models\neg\phi_i$.
\end{itemize}
The next step is to show that the set of possible extensions of $\pp$
to $\pp'$ with $\pp'\models\bigwedge RE\land\bigwedge DI$ is definable.
Let $W_{i}$ be a concept which contains all possible memory cells
accessed during the run of $S_{i}$. E.g. for $S_2$ (see Fig. \ref{fig:list-example}),
\begin{eqnarray*}
W_{2} & \equiv & Z_{I}\sqcup Z_{II}\sqcup Z_{III}\\
Z_{I} & \equiv & b\\
Z_{II} & \equiv & \exists next^{-}.b\\
Z_{III} & \equiv & \exists next^{-}.\exists next^{-}.b
\end{eqnarray*}
$W_{2}$ can be extracted naively from $S_{2}$. The remaining elements
of the memory at the start of $S_{2}$ are two lists segments
\begin{eqnarray*}
Z_{IV} & \equiv & L_{b}\sqcap\neg W_{2}\\
Z_{V} & \equiv & L\sqcap\neg L_{b}\sqcap\neg W_{2}
\end{eqnarray*}
which are guaranteed not to be accessed during the run of $S_{2}$.

\begin{figure}
\begin{center}
\includegraphics[scale=0.7]{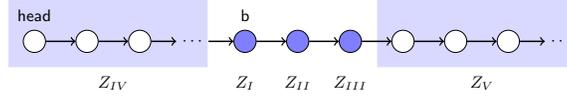}
\end{center}
\caption{\label{fig:list-example} The figure depicts the input of one iteration
of the loopless code $S_{2}$ inside the loop, namely a list starting
at $head$ with $b$ pointing to an element in the list. $S_{2}$
only accesses $b$, $b.next$ and $b.next.next$. Therefore, after
the execution of $S_{2}$, $Z_{IV}$ and $Z_{V}$ appear in any list
segment as entirely or not at all. }
\end{figure}

Since the elements of $Z_{IV}$ are not accessed and their pointers remain unchanged, they agree on which
list segments they belong to in $\mm_{S_{i}}$, and similarly for
$Z_{V}$. In other words, every list segment in $\mm_{S_{i}}$ consists
of a concatenation of a subset of the elements and list segments $\{Z_{I},\ldots,Z_{V}\}$.
I.e. every list segment in $\mm_{S_{i}}$ can be characterized entirely
by whether, and in what order, $Z_{I},\ldots,Z_{V}$ occur in it.

Let $\sigma=\{C_{y}\mid y\subseteq\{Z_{I},\ldots,Z_{V}\}\}$ be a
vocabulary consisting of fresh concepts. For $y\subseteq\{Z_{I},\ldots,Z_{V}\}$,
let $\alpha_{y}=\bigsqcup_{Z_{s}\in y}Z_{s}\equiv C_{y}$. The formula
\[
\alpha=\bigwedge_{y\subseteq\{Z_{I},\ldots,Z_{V}\}}\alpha_{y}
\]
 defines the $C_{y}$ to be all possible candidates for list segments
in $\mm_{S_{i}}$.

A concept $C_{y}$ is a list segment if there is a linear ordering
witnessing the concatenation of the elements and list segments in
$y$ in $\mm_{S_{i}}$. Let $\beta_{y}^{o}$ be the disjunction over
all linear orderings $\leq$ of $y$ with minimal element $min_{\leq}$
of $\beta_{\leq}^{o}$, where
\begin{eqnarray*}
\beta_{\leq}^{o} & = & min_{\leq}\equiv o\land\bigwedge_{Z_{s}\leq Z_{t}\in y}\bot\not\equiv(\exists\underline{next}^{-}.Z_{s})\sqcap Z_{t}
\end{eqnarray*}
$\beta_{\leq}^{o}$ expresses that the elements and list segments
in $y$ are concatenated according to $\leq$ starting from $o$,
and $\beta_{y}^{o}$ expresses that such an ordering exists. E.g.,
$\alpha_{y}\land\beta_{y}^{\underline{head}}$ holds iff $C_{y}$
consists exactly of $\bigsqcup_{Z_{s}\in y}Z_{s}$ and is a list segment
starting at $\underline{head}$. Let $Y$ be the set of pairs $(y_{1},y_{2})$
of disjoint subsets of $\{Z_{I},\ldots,Z_{V}\}$. (C) is equivalent
to:
\begin{itemize}
\item [(D)] There exist a $\tau$-structure $\nn$ such that $\nn\models\Phi_{start}^{2}$
and we have
\[
\left\langle \nn,\nn_{S_{2}}\right\rangle \models \gamma
\]
where $\gamma$ is
\[
\alpha\land\bigwedge_{(y_{1},y_{2})\in Y}\left(\left(\beta_{y_{1}}^{\underline{head}}\land\beta_{y_{2}}^{\underline{b}}\right)\to\neg\phi_i[\underline{L_{b}}\backslash C_{y_{1}},\underline{L}\backslash C_{y_{1}}\sqcup C_{y_{2}}]\right)
\]
The notation $\phi_i[A\backslash B]$ denotes the syntactical substitution
of $A$ with $B$ in $\phi_i$.
\end{itemize}
$\gamma$ expresses that for every two list segments starting in $\underline{head}$
and $\underline{b}$ respectively, $\phi_i$ does not hold.

Now we can turn to the symbols of $\underline{\rho}$. We can apply
backwards propagation to $head$, $b$, $\marked$ and $next$, i.e.
we can compute the weakest precondition predicate transformer for
each of these with respect to $S_{i}$. For $next$ we apply the transformer
directly on concepts $\exists next.C$ and $\exists next^{-}.C$ for
any $C$ which use $next$.

In \cite{arXiv:CKSVZ13} the following is proven\footnote{The description logic used in \cite{arXiv:CKSVZ13} is slightly more powerful, allowing also role inclusion. However, this is easy to overcome, see footnote \ref{foot-sub}. }:
\begin{theorem}
\label{th:back} For every disjoint vocabularies $\xi$ and $\tau$,
every $\Loiq(\tau\cup\underline{\tau}\cup\xi)$-formula $\phi_i$,
and loopless code $S$, it is possible to compute
a $\Loiq(\rho\cup(\underline{\tau}\backslash\underline{\rho})\cup\xi)$-formula
$\theta_{S,\phi_i}$  such that $\pp\models\theta_{S,\phi}$
iff $\pp'=\left\langle \pp,\pp_{S}\right\rangle \models\phi_i$.
\end{theorem}

We apply the backwards propagation to $\gamma$.

For $S = S_{2}$,
let $\phi_{false}$ be obtained from $\gamma$
by syntactically substituting $b$ to $\exists next^{-}.b$.
Let $\phi_{true}$ be obtained from $\gamma$ by
syntactically substituting\footnote{\label{foot-sub} In \cite{arXiv:CKSVZ13} $next$ would have been subtituted syntactically for an expression such as $next \setminus (b \times \top)  \cup (\exists next^-.b \times \exists next^-.\exists next^-.b)$. Since we do not allow role inclusions and $\times$ in out logic, we make the substitutions on the concepts which use $next$ instead.  }
$\exists next.D$ to
\[\stackrel{(\bullet)}{\overbrace{(b\sqcap\exists next.\exists next.D)}}\sqcup\stackrel{(\bullet \bullet)}{\overbrace{(\neg b\sqcap\exists next.D)}}\]
and $\exists next^-.D$ to
\[
 \stackrel{(\bullet)}{\overbrace{\left(\exists next^-.\exists next^-.(b\sqcap D)\right)}} \sqcup \stackrel{(\bullet \bullet)}{\overbrace{\left( \exists next^-.(D \sqcap \neg b)\right)}}
\]
for every concept $D$.
$(\bullet)$ expresses the new value of $next$ on $b$. $(\bullet \bullet)$ expresses the unchanged values of $next$.
Then
\begin{eqnarray*}
\theta_{S_{2},\gamma} & = & (\phi_{cond}\land\phi_{true})\lor(\neg\phi_{cond}\land\phi_{false})\\
\phi_{cond} & = & \exists next^{-}.b\sqsubseteq\neg null\sqcap \marked
\end{eqnarray*}

As a consequence we get that (D) is equivalent to (E):
\begin{itemize}
\item [(E)] There exist a $\tau$-structure $\nn$ such that $\nn\models\Phi_{start}^{2}\land\theta_{S_{2},\gamma}$.
\end{itemize}
Whether (E) holds reduces to the finite satisfiablity problem of $\L$.

In \cite{arXiv:CKSVZ13} a functional program analysis based on Theorem
\ref{th:back} was discussed. The memory model of \cite{arXiv:CKSVZ13}
is considerably more developed in order to allow allocation and deallocation
of memory. However the step between (D) and (E) was not considered
there, leading to a program analysis which is based on an existing shape analysis.
Since the goal of this section is not to develop a program analysis
but rather give the reader intuition via an example, we did not strive
to make the formulas here most efficient or small. We will consider
the development of a full program analysis based on this logic in
future work.

\paragraph{Content analysis}

The program in Fig. \ref{fig:prog} fits naturally to the type of
content analysis of \cite{arXiv:CKSVZ13}. Consider the information
system of a hotel. A partial UML of the system is depicted in Fig.
\ref{fig:UML-hotel}. The information system of the hotel contains
data about rooms, guests, bookings, payments, personnel, etc. A simple
implementation of the system from Fig. \ref{fig:UML-hotel} may contain
three disjoint lists for the rooms, the guests and the active bookings.
\cite{arXiv:CKSVZ13} shows how to express UML-like content invariants
in the context of programs with dynamic data structures.
\begin{figure}
\begin{center}\includegraphics{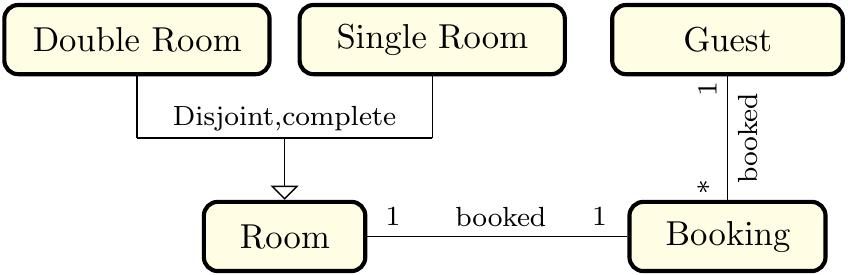}\end{center}

\caption{\label{fig:UML-hotel} }
\end{figure}

When a guest who has a booking has not checked-in until the next morning,
they are charged for one night and the rest of their booking is canceled.
The program in Fig. \ref{fig:prog2} is executed once a day before
the reception counter is opened. The program traverses the list $Bk$
of active bookings to search for non-arrivals. When such a booking
is found, it is removed from the list and the function $\mathsf{no-arrival}$,
which charges the guest for one night, is run.

\begin{figure}
\begin{lstlisting}
start: b := hBk;
loop:  while (b != null) {
         if (b.next != null &&
                       !b.next.chkd_in)
         {
            non-arrival(b);
            b.next := b.next.next;
         }
         else
            b := b.next;
        }
end:    ;
\end{lstlisting}

\caption{\label{fig:prog2}}
\end{figure}
In this case, the pre-condition is
\[
\Phi_{pre}=\Phi_{aList}^{hBk,next,Bk}\land hBk\sqsubseteq CheckedIn
\]

and the post-condition is given by $\Phi_{post}=\Phi_{post1}\land\Phi_{post2}$
where

\begin{eqnarray*}
\Phi_{post1} & = & \Phi_{aList}^{hBk,next,Bk}\land Bk\sqsubseteq CheckedIn\\
\Phi_{post2} & = & Bk\equiv Bk_{ghost}\sqcap CheckedIn_{ghost}
\end{eqnarray*}
The correctness proof for the program in Fig. \ref{fig:prog2} is
the same as the correctness proof of the program in Fig. \ref{fig:prog}.

\section{\label{se:main-section} Decision procedures for $\L$}

Let $\Phi=\varphi\land\bigwedge RE\land\bigwedge DI$ be an $\L$-formula.
Let $\tau=\NC\cup\NR\cup\NI$ be the vocabulary of $\varphi$.
We denote $RE=\{\Reacha{B_{1}}{S_{1}}{A_{1}},\ldots,\Reacha{B_{h}}{S_{h}}{A_{h}}\}$.

In this section we prove the existence a decision procedure of finite satisfiability and implication for $\L$:

\begin{thm-th:exp}
Let $\Phi_{i}=\varphi_{i}\land\bigwedge RE_{i}\land\bigwedge DI_{i}\in\L$, for $i=1,2$.
There are $\LoiqNOb$ formulas $\mu$ and $\kappa$ over an extended vocabulary such that\\
(1) $\Phi_{1}$ is finitely satisfiable iff $\mu$ is finitely satisfiable.\\
(2) $\Phi_{1}$ implies $\Phi_{2}$ iff $\kappa$ is not finitely satisfiable.
\end{thm-th:exp}

\paragraph{Outline of proof. }

$\formula{\Phi}$ already belongs to $\Loiq$.
The reachability requirements are missing in order to capture $\Phi$.
The models of $\formula{\Phi}$ can be partitioned into \emph{standard }and \emph{non-standard models}, depending on whether they satisfy $\bigwedge RE$.
In general, we cannot augment $\formula{\Phi}$ in $\Loiq$ to eliminate the non-standard models, since reachability is not expressible
in $\Loiq$.
However, we can augment it so that it is guaranteed that whenever a non-standard model exists, so does a standard model.
To do so, we define \emph{semi-connectedness}, which is a weaker requirement than satisfying $\bigwedge RE$.
A model is semi-connected if every element of its universe which should be reachable according to some $\Reacha{B_{i}}{S_{i}}{A_{i}}$ and is not, is reachable from a cycle in $A_{i}$.
We show that semi-connectedness is expressible in $\Loiq$.

Under certain conditions, it is possible to repeatedly apply an operation $\rhd$, which turns non-standard but semi-connected models into standard models, by eliminating the said cycles.
The existence of a non-standard semi-connected model then implies the existence of a standard model.
A sufficient condition under which semi-connected models can be turned to standard models using $\rhd$ is that the elements in $A_{i}$ admit so-called \emph{useful labelings}.
Useful labelings mimic an order relation on the types of the elements in $A_{i}$ and guarantee that applying the operation $\rhd$ makes progress towards a standard model.
We show that having useful labelings is expressible in $\Loiq$.

As a consequence we get a decision procedure for satisfiability of
$\Phi$, which amounts to adding to $\formula{\Phi}$ the requirements that models are semi-connected and have useful labelings.
The resulting $\LoiqNOb$-formula is satisfiable iff $\Phi$ is.
A decision procedure for implication is obtained as consequence.
Decision procedures which are tight in terms of complexity are given in Section \ref{se:poly}.
In Section \ref{se:reduce-to-formula} we give simpler but complexity-wise suboptimal decision procedures.
The decision procedures in Section \ref{se:poly} follow the same plan, and differ only in the construction and sizes of the formulae expressing the existence of useful labelings.

\subsection{Types and the operation $\rhd$}

We write $C\in\varphi$ if there exists a concept $D$ and an inclusion
$C\sqsubseteq D$ or $D\sqsubseteq C$ which occurs in $\varphi$.
$C$ and $D$ need not be atomic.
\begin{definition}
[$\rhd$] Let $\mm$ be a $\tau$-structure.
Let $a_{0},a_{1}\in M$ and $r\in\NF$ and $\mathfrak{t} = (a_{0},a_{1},r)$.
Let $\op{\mm}{\mathfrak{t}}$ be the structure such that $\mm$ and $\op{\mm}{\mathfrak{t}}$ have the same universe $M$ and the same interpretations of every atomic concept, nominal and atomic role except for $r$, and $r^{\op{\mm}{\mathfrak{t}}} = r^{\mm} \setminus \{(a_{i},b) \mid (a_{i},b) \in r \text{ and } i \in \{0,1\}\} \cup$ $\{(a_{1-i},b) \mid (a_{i},b) \in r \text{ and } i \in \{0,1\}\}$.
\end{definition}

For the main property of $\rhd$ we need the notion
of types:
\begin{definition}
[Types]
We define $\TYPES_{\varphi} = \textbf{2}^{ \{C \mid C \in \varphi\} }$ as the powerset over the set of concepts appearing in $\varphi$.
Let $\mm$ be a $\tau$-structure $\mm$ and $u \in M$.
We denote by $\ovtpn_{\mm}^{\varphi}(u) \in \TYPES_{\varphi}$ the set of concepts $C \in \varphi$ such that $u \in C^{\mm}$.
We call $\ovtpn_{\mm}^{\varphi}(u)$ the \emph{type} of $u$.
We sometimes omit the subscript $\mm$ when it is clear from the context.
\end{definition}
We note that the size of $\TYPES_{\varphi}$ is at most $2^{|\varphi|}$.

\begin{lemma}
\label{lem:types}
\label{lem:1-type-to-2-type-a-}
\label{lem:1-type-to-2-type-b-}
Let $\mm_{1}$ and $\mm_{2}$ be two $\tau$-structures with the same universe $M$.
If for all $u\in M$ we have $\ovtpn_{\mm_{1}}^{\varphi}(u)=\ovtpn_{\mm_{2}}^{\varphi}(u)$,
then $\mm_{1}$ and $\mm_{2}$ agree on $\varphi$.
\end{lemma}

\begin{proof}
$\varphi$ is a Boolean combination of inclusion assertions.
Therefore, it is enough to show $\mm_{1}\models C\sqsubseteq D$ iff $\mm_{2}\models C\sqsubseteq D$ for all of the inclusion assertions $C\sqsubseteq D$ which occur
in $\varphi$.
Let $C\sqsubseteq D$ be such an inclusion assertion.
For $u\in M$, $u\in C^{\mm_{1}}$ iff $u\in C^{\mm_{2}}$, and $u\in D^{\mm_{1}}$
iff $u\in D^{\mm_{2}}$.
Hence, $C^{\mm_{1}}=C^{\mm_{2}}$ and $D^{\mm_{1}}=D^{\mm_{2}}$,
implying $\mm_{1}\models C\sqsubseteq D$ iff $\mm_{2}\models C\sqsubseteq D$.
\end{proof}

The crucial property of $\rhd$ is that $\mm$ and $\op{\mm}{\mathfrak{t}}$ agree on $\varphi$ if $a_{0}$ and $a_{1}$ have the same type: 

\begin{lemma}
\label{lem:op-on-concepts} \label{lem:mm-op--}
Let $\mm$ be a $\tau$-structure, let $a_{0}, a_{1} \in M$ such that $\ovtpn_{\mm}^{\varphi}(a_{0}) = \ovtpn_{\mm}^{\varphi}(a_{1})$,
let $r \in \NF$, and let $\mathfrak{t}=(a_{0},a_{1},r)$.
\begin{enumerate}[(1)]
 \item $C^{\mm}=C^{\op{\mm}{\mathfrak{t}}}$ for all $C\in\varphi$.
\end{enumerate}
Consequently:
\begin{itemize}
 \item[(2)]  For every $u\in M$, $\ovtpn_{\mm}^{\varphi}(u) = \ovtpn_{\op{\mm}{\mathfrak{t}}}^{\varphi}(u)$.
 \item[(3)] $\mm\models \varphi$ iff $\op{\mm}{\mathfrak{t}}\models \varphi$.
 \end{itemize}
\end{lemma}

Statement (1) of Lemma \ref{lem:op-on-concepts} is proven by induction on the construction of concepts in $\varphi$.
(2) follows directly from (1). 
(3) follows using Lemma \ref{lem:types}.
See Appendix \ref{app:proofs-lem:op-on-concepts} for a detailed proof.

\subsection{Semi-connectedness and useful labelings}

Here we define semi-connectedness and useful labelings exactly and prove that they capture reachability (Lemma~\ref{lem:connected-implies-useful-order--}).

\begin{definition}
[Semi-connected Structure]
\label{def:Po-conn}
For every reachability assertion $\beta_{h'} = \Reacha{B_{h'}}{S_{h'}}{A_{h'}}$, we write $D_{h'}^{\mm}$ for \emph{the directed graph $D_{\beta_{h'}}^{\mm}$}.
Let $\mm$ be a $\tau$-structure.
$\mm$ is $\Phi$-semi-connected, if
(I) $\mm \models \formula{\Phi}$ and
(II) for every $\Reacha{B_{h'}}{S_{h'}}{A_{h'}}\in RE$ and $u\in A_{h'}^{\mm}$, either $u$ is reachable in $D_{h'}^{\mm}$ from $B_{h'}^{\mm}$ or $u$ is reachable from a cycle.
\end{definition}

Observe that if $\mm$ is $\Phi$-semi-connected, then $\mm\models\bigwedge RE\land\bigwedge DI$
iff $\mm$ satisfies the following strengthening of {\it (II)}:
{\it for every $\Reacha{B_{h'}}{S_{h'}}{A_{h'}}\in RE$ and $u\in A_{h'}^{\mm}$, $u$ is reachable from $B_{h'}^{\mm}$}.
The $h'$-useful labelings we define next mimic linear orderings on the types of the elements in $A_{h'}^{\mm}$ that can be obtained from a Depth-First Search (DFS) run on $D_{h'}^{\mm}$ starting from elements in $B_{h'}^{\mm}$.
\begin{definition}[Useful Labeling]
\label{def:useful-labeling}
Let $\mm$ be a $\tau$-structure.
Let $1\leq h'\leq h$.
A function $f_{h'}: A_{h'}^{\mm} \rightarrow [1,|\TYPES_{\varphi}|]$ is a \emph{$h'$-useful labeling for $\mm$}, if (1) $f_{h'}(u) = f_{h'}(v)$ implies $\ovtpn_{\mm}^{\varphi}(u) = \ovtpn_{\mm}^{\varphi}(v)$ for all $u,v \in A_{h'}^{\mm}$ and if (2) for every element $u\in A_{h'}^{\mm}$, either $u \in B_{h'}^{\mm}$,
or there exist elements $v,w\in A_{h'}^{\mm}$ such that $f_{h'}(u) = f_{h'}(v)$, $f_{h'}(w) < f_{h'}(v)$ and the graph $D_{h'}^{\mm}$ has an edge $(w,v)$.
\end{definition}

\begin{lemma}
\label{lem:use-ord}
Let $\mm$ be a $\tau$-structure.
If $\mm \models \Phi$, then there are $h'$-useful labelings $f_{h'}$ for $\mm$, for every $1\leq h'\leq h$.
\end{lemma}
\begin{proof}
We assume $\mm$ satisfies $\Phi = \varphi \land \bigwedge RE \land \bigwedge DI$.
We fix some $1\leq h'\leq h$.
We have that $\Reacha{B_{h'}}{S_{h'}}{A_{h'}}$ holds for $\mm$.
We explain how to build a $h'$-useful labeling for $\mm$ by executing a Depth-First Search (DFS) from the elements in $B_{h'}^{\mm}$.
If an element $u$ is visited during the DFS, $u$ is assigned a number according to its type $\ovtpn_{\mm}^{\varphi}(u)$.
If the type $\ovtpn_{\mm}^{\varphi}(u)$ has not appeared during the DFS yet, $u$ is assigned the smallest number in $[1,|\TYPES_{\varphi}|]$ that has not been used so far;
if the type has already appeared, $u$ is assigned the number associated with this type.
Let $f_{h'}$ be the labeling resulting from this process.
We show that $f_{h'}$ is a $h'$-useful labeling for $\mm$.

By construction $f_{h'}$ is a function from $A_{h'}^{\mm}$ to $[1,|\TYPES_{\varphi}|]$.
Moreover, for all $u,v \in A_{h'}^{\mm}$ is holds that $f_{h'}(u) = f_{h'}(v)$ iff $\ovtpn_{\mm}^{\varphi}(u) = \ovtpn_{\mm}^{\varphi}(v)$ (*).
It remains to show that for every element $u\in A_{h'}^{\mm}$, either $u \in B_{h'}^{\mm}$,
or there exist elements $v,w\in A_{h'}^{\mm}$ such that $f_{h'}(u) = f_{h'}(v)$, $f_{h'}(w) < f_{h'}(v)$ and the graph $D_{h'}^{\mm}$ has an edge $(w,v)$:
Let $u \in A_{h'}^{\mm} \setminus B_{h'}^{\mm}$.
We proceed by a case distinction:
(1) The type $\ovtpn_{\mm}^{\varphi}(u)$ of $u$ has not been seen during the DFS before $u$ is visited.
Because of $u \not\in B_{h'}^{\mm}$ there is a predecessor $w \in A_{h'}^{\mm}$ of $u$ in $D_{h'}^{\mm}$ through which $u$ has been reached during the DFS.
Because $w$ has been reached before $u$ and because $u$ is assigned the smallest number in $[1,|\TYPES_{\varphi}|]$ that has not been used so far, we have $f_{h'}(w) < f_{h'}(u)$.
(2) The type $\ovtpn_{\mm}^{\varphi}(u)$ of $u$ has already been seen during the DFS before $u$ is visited.
Let $v \in A_{h'}^{\mm}$ be the first node with type $\ovtpn_{\mm}^{\varphi}(v) = \ovtpn_{\mm}^{\varphi}(u)$ to be visited during the DFS.
By case (1) there is a predecessor $w$ of $v$ in $D_{h'}^{\mm}$ with
$f_{h'}(w) < f_{h'}(v)$.
By (*) we have $f_{h'}(u) = f_{h'}(v)$.
Thus, the claim follows.
\end{proof}


This gives direction $\Rightarrow$ of the following lemma:
\begin{lemma}
\label{lem:connected-implies-useful-order--} \label{lem:conn-equiv-semi-conn-plus-useful-ord--}
$\Phi=\varphi\land\bigwedge RE\land\bigwedge DI$ is satisfiable iff
there is a $\Phi$-semi-connected structure with $h'$-useful labelings for every $1\leq h'\leq h$.
\end{lemma}

Next, we introduce definitions that will be needed for the proof of direction $\Leftarrow$.

Let $G=(V,E)$ be a directed graph.
$\reach{G}{X} = Y$ denotes the set of  elements $Y \subseteq V$ that are \emph{reachable} from $X \subseteq V$ in $G$.

\begin{definition}
[Base and Values]
Let $f$ be a $h'$-useful labeling for $\mm$.
We call a set $X \subseteq A_{h'}^\mm$ a \emph{base} for $D_{h'}^{\mm}$, if
$\reach{D_{h'}^{\mm}}{X} = A_{h'}^\mm$.
We call a member $x$ of a base $X$ a \emph{base element}.
We define the \emph{value} $\val_f(X) = \sum_{x \in X \setminus B_{h'}^\mm} f(x)$ \emph{of a base} $X$ to be the sum over the label values of the base elements of $X$ that are not in $B_{h'}^\mm$.
We define the \emph{value} $\val_f(D_{h'}^{\mm}) = \min \{\val_f(X) \mid X \text{ is a base for } D_{h'}^{\mm} \}$ \emph{of the graph} $D_{h'}^{\mm}$ to be the minimum of the values of its bases.
We omit the subscript $f$ in $\val(X)$ and $\val(D_{h'}^{\mm})$ when $f$ is clear from the context.
\end{definition}

Intuitively, values $\val_f(D_{h'}^{\mm})$ capture how close the graph $D_{h'}^{\mm}$ is to being connected:

\begin{lemma}
\label{lem:value-zero}
Let $f$ be a $h'$-useful labeling for $\mm$.
$\val_f(D_{h'}^{\mm}) = 0$ iff $D_{h'}^{\mm}$ is connected.
\end{lemma}
\begin{proof}
Assume $D_{h'}^{\mm}$ is connected.
Then $B_{h'}^\mm$ is a base for $D_{h'}^{\mm}$.
Thus $\val(B_{h'}^\mm) = 0$, which implies $\val(D_{h'}^{\mm}) = 0$.

Assume $\val(D_{h'}^{\mm}) = 0$.
Then there is a base $X$ for $D_{h'}^{\mm}$ with $\val(X) = 0$.
Because $f$ maps all nodes to positive values, we must have $X \subseteq B_{h'}^\mm$.
Thus, every node in $D_{h'}^{\mm}$ is reachable from $B_{h'}^\mm$.
\end{proof}

The following lemma states a property of bases in semi-connected structures:

\begin{lemma}
\label{lem:basis-circle}
Let $\mm$ be a structure that is $\Phi$-semi-connected.
Let $f$ be a $h'$-useful labeling for $\mm$.
Let $X$ be a base for $D_{h'}^{\mm}$ with $\val_f(X) = \val_f(D_{h'}^{\mm})$.
Then every base element $x \in X$ either belongs to $B_{h'}^\mm$ or to a cycle of $D_{h'}^{\mm}$.
\end{lemma}
\begin{proof}
We fix some base element $x \in X$.
Let us assume $x$ does not belong to $B_{h'}^\mm$ or to a cycle of $D_{h'}^{\mm}$.
By definition we have $\val(X) = \val(D_{h'}^{\mm}) = \min \{\val(X) \mid X \text{ is a base for } D_{h'}^{\mm} \}$.
By the semi-connectedness of $\mm$, $x$ is either reachable from $B_{h'}^\mm$ (1) or from a cycle in $D_{h'}^{\mm}$ (2).
Case (1): $x$ is reachable from some $z \in B_{h'}^\mm$.
However, $X' = X \setminus \{x\} \cup \{z\}$ is a base for $D_{h'}^{\mm}$ with $\val(X') < \val(X)$.
Contradiction.
Case (2): $x$ is reachable from some cycle $C$ in $D_{h'}^{\mm}$.
We fix an element $y$ on $C$.
Because $X$ is a base, there is a base element $z \in X$ such that $y$ can be reached from $z$.
We have $z \neq x$, because otherwise $x$ would belong to a cycle of $D_{h'}^{\mm}$.
However, $X' = X \setminus \{x\}$ is a base for $D_{h'}^{\mm}$ with $\val(X') < \val(X)$.
Contradiction.
\end{proof}

The next lemma, Lemma \ref{lem:rhd-connectedness--}, shows that $\rhd$ can be applied to $D_{h'}^{\mm}$ such that $\val(D_{h'}^{\mm})$ decreases for some $1 \le h' \le h$.
The compatibility of $RE$ and $DI$ ensures that $\rhd$ does not modify the graphs $D_\ell^{\mm}$ with $\ell \neq h'$.

\begin{lemma}
\label{lem:rhd-connectedness--}
Let $1\leq h' \leq h$ and let $\mm$ be a structure such that $\mm \models \formula{\Phi}$, $\mm$ is $\Phi$-semi-connected and has $\ell$-useful labelings $f_\ell$ for all $1 \leq \ell \leq h$.
If $\val_{f_{h'}}(D_{h'}^{\mm}) > 0$,
then there is a tuple\textup{ $\mathfrak{t}=(a_{0},a_{1},r)$
}such that
\begin{enumerate}
\item For all $\ell \not= h'$, $D_\ell^{\mm} = D_\ell^{\op{\mm}{\mathfrak{t}}}$.
\item For all $u\in M$, $\ovtpn_{\mm}^{\varphi}(u)=\ovtpn_{\op{\mm}{\mathfrak{t}}}^{\varphi}(u)$.
\item $\op{\mm}{\mathfrak{t}}\models \formula{\Phi}$.
\item $\val_{f_{h'}}(D_{h'}^{\mm}) > \val_{f_{h'}}(D_{h'}^{\op{\mm}{\mathfrak{t}}})$.
\item $\val_{f_{\ell}}(D_\ell^{\mm}) = \val_{f_{\ell}}(D_\ell^{\op{\mm}{\mathfrak{t}}})$
for all $\ell \not= h'$.
\item $\op{\mm}{\mathfrak{t}}$ is $\Phi$-semi-connected.
\item $f_\ell$ is a $\ell$-useful labeling for $\op{\mm}{\mathfrak{t}}$
for all $1\leq \ell \leq h$.
\end{enumerate}
\end{lemma}
\begin{proof}[Proof Sketch]
In the following, we will give an intuition on the proof of Lemma~\ref{lem:rhd-connectedness--}. We delay the full proof to
Section \ref{app:useful-proof}.
We fix some base $X$ for $D_{h'}^{\mm}$ with $\val(X) = \val(D_{h'}^{\mm}) > 0$.
We choose a base element $a_1 \in X \setminus B_{h'}^\mm$.
Because $f_{h'}$ is a $h'$-useful labeling for $\mm$ there are $a_0,w\in A_{h'}^{\mm}$ such that $f_{h'}(a_0) = f_{h'}(a_1)$, $f_{h'}(w) < f_{h'}(a_0)$ and the graph $D_{h'}^{\mm}$ has an edge $(w,a_0)$.
$a_0$ and $a_1$ cannot belong to the same cycle in $D_{h'}^{\mm}$ by the minimality of $X$.
By Lemma~\ref{lem:basis-circle}, $a_1$ belongs to a cycle in $D_{h'}^{\mm}$.
$b_1$ denotes the successor of $a_1$ by some edge $r$ in this cycle.
\piccaption[]{\label{fig:1}}\parpic[r]{\includegraphics[scale=0.85]{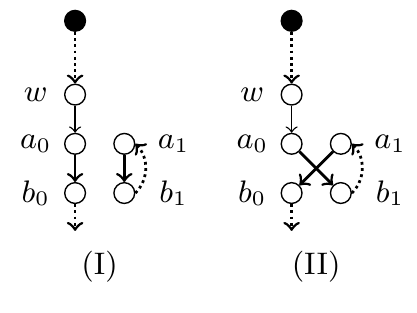}}
Figure \ref{fig:1} shows the result (II) of applying $\rhd$ on (I) (we depict here the case, where $a_0$ has an $r$-successor $b_0$).
The black vertex belongs to the base $X$, dotted arrows denote paths inside $A_{h'}^{\mm}$ whose edges belong to $S_{h'}$, and solid arrows belong to $S_{h'}$.
Applying $\rhd$ increases the reachability of the structure: all
vertices in (II) are now reachable from the black vertex.
However, in the special case where the black vertex and $a_1$ coincide, a new cycle in created.
In both cases we have that $X' = X \setminus \{a_1\} \cup \{w\}$ is a base for $D_{h'}^{\op{\mm}{\mathfrak{t}}}$ with $\val(X') < \val(X)$ and that  $D_{h'}^{\op{\mm}{\mathfrak{t}}}$ remains semi-connected.
\end{proof}

Finally, we show that the repeated application of $\rhd$ on a semi-connected structure with useful labelings leads eventually to a structure satisfying $\Phi$:

\begin{proof}[Proof of Lemma~\ref{lem:connected-implies-useful-order--}]
By Lemma~\ref{lem:use-ord} if $\Phi$ is satisfied by a structure $\mm$, then $\mm$ has $h'$-useful labelings for every $h'$ and is semi-connected.

For the other direction, there is a $\Phi$-semi-connected structure $\mm$ with $h'$-useful labelings for every $1\leq h'\leq h$.
There is a sequence $\mm = \mm_{1},\ldots,\mm_{p} = \mm'$ of structures such
that each $\mm_{j+1}$ is obtained from $\mm_{j}$ by one application
of $\rhd$ and $\val(D_{h'}^\mm)=0$ for all $h'$.
There is guaranteed to be such a sequence because
\begin{enumerate}
\item the premise of Lemma \ref{lem:rhd-connectedness--} holds for $\mm$,
\item for all $j$, if the premise of Lemma \ref{lem:rhd-connectedness--}
holds for $\mm_{j}$ and $\phi$, then the premise of Lemma \ref{lem:rhd-connectedness--}
holds for $\mm_{j+1}$, and
\item The tuple $(\val(D_1^\mm) ,\cdots,\val(D_h^\mm))$ is decreasing with regard to the component-wise ordering of $h$-tuples over $\mathbb{N}$, so eventually $(0,\ldots,0)$ must be reached.
\end{enumerate}
By Lemma~\ref{lem:value-zero}, $D_{h'}^{\mm'}$ is connected for all $h'$.
\end{proof}

\subsection{Reducing connected satisfiability to (plain finite) satisfiability
\label{se:reduce-to-formula}}

Here we show how to express semi-connectedness and existence of useful labelings in $\Loiq$.
For semi-connectedness, this is easy:
\begin{lemma}
\label{lem:semi-def}
There exists a formula $\semi{\Phi} \in \Loiq$ such that $\mm \models \semi{\Phi}$ iff $\mm$ is $\Phi$-semi-connected.
\end{lemma}
\begin{proof}
Let $\semi{\Phi} = \formula{\Phi} \land \bigwedge_{1\leq h'\leq h} \delta_\reachCyclic^{h'}$
and $\delta_\reachCyclic^{h'} = A_{h'}\sqcap\neg B_{h'} \sqsubseteq \bigsqcup_{s\in S_{h'}} \exists s^{-}.A_{h'}$.
Assume $\mm$ is $\Phi$-semi-connected.
We fix some $1 \le h' \le h$.
Every $u\in A_{h'}^{\mm}$ is reachable from $B_{h'}^{\mm}$ or is reachable from a cycle, and therefore $u$ has a predecessor in $A_{h'}^{\mm}$ with respect to $\bigcup_{s\in S_{h'}}s^{\mm}$, unless $u\in B_{h'}$, so $\mm \models \delta_\reachCyclic^{h'}.$

Conversely, we consider a structure $\mm$ with $\mm \models \formula{\Phi}$, but which is not $\Phi$-semi-connected.
There exists a vertex $u$ and a $1 \le h' \le h$ such that $u$ is not reachable from $B_{h'}^{\mm}$ nor from a cycle.
There must exist a vertex $v$ in $A_{h'}^{\mm}$ which is a predecessor of $u$ (possibly $u$ itself) and which does not have a predecessor in $A_{h'}^{\mm}$, otherwise $u$ lies on a cycle (using the finiteness of the universe).
Since $u$ is reachable from $v$, we must have that, like $u$, $v$ not reachable from $B_{h'}^{\mm}$.
Therefore, $v$ is not in $B_{h'}^{\mm}$ but belongs to $A_{h'}^{\mm}$
and does not have a predecessor, so $\mm \not \models \delta_\reachCyclic^{RE,DI}$.
\end{proof}

Next we define a set of structures $\ORDW(\varphi)$ that represent models of $\varphi$ and at the same time also contain useful labelings.
After this definition we will show that $\ORDW(\varphi)$ can be defined inside the logic $\Loiq$.

\begin{definition}
\label{def:ORD}
Let $\varphi$ be a $\Loiq$ formula over vocabulary $\tau$.
Let $k = |\TYPES_{\varphi}|$.
We define an extended vocabulary $\useful(\tau)$ that extends $\tau$ with a new atomic concept $M$, new nominals $o_1,\ldots,o_k$,
a new atomic role $\order$ and new functional atomic roles $f_1,\ldots, f_h$.

Let $\nn$ be a $\useful(\tau)$-structure with universe $N$.
We denote the substructure of $\nn$ with universe $M^\nn$ by $\mm$.
We denote the set $N \setminus M^\nn$ by $\nominalsExp^\nn$.
The structure $\nn$ belongs to $\ORDW(\varphi)$ if the following conditions hold:
\begin{enumerate}
\item \label{prop:Mtilde} $\mm$ satisfies $\varphi$.
\item \label{prop:partition} $N$ is partitioned into $M^\nn$ and $\nominals^\nn = \{o_1^\nn,\ldots,o_k^\nn\}$.
\item \label{prop:linear} We have that $(o_i^\nn,o_j^\nn) \in \order^\nn$ iff $i < j$.
\item \label{prop:ftypes} $f_{h'}^\nn$ is a function from $A_\ell^\mm$ to $\nominals^\nn$, for every $1 \le {h'} \le h$.
\item \label{prop:ford} $f_{h'}^\nn$ is a $h'$-useful labeling for $\mm$, using $ \nominalsExp^\nn$ for the natural numbers $[1,k]$ and $\order$ for the order on the natural numbers in Definition~\ref{def:useful-labeling}, for every $1 \le h' \le h$.
\end{enumerate}
\end{definition}

\begin{lemma}
\label{lem:extension-useful}
$\ORDW(\varphi)$ is non-empty iff there is a model $\mm$ of $\varphi$ with $h'$-useful labelings for $\mm$ for every $1\leq h'\leq h$.
\end{lemma}
\begin{proof}
Let $\mm$ be a model of $\varphi$ with $h'$-useful labelings $f'_{h'}$ for every $1 \leq h'\leq h$.
Let $M'$ be the universe of $\mm$.
We define a model $\nn$ with universe $N := M' \cup [1,k]$ by
\begin{itemize}
\item $M^\nn = M'$,
\item $o_i^\nn := i$,
\item $\order^\nn = \{(i,j) \mid 1 \le i < j \le k\}$,
\item $f_i^\nn := f'_i$, and
\item $C^\nn = C^\mm$ for all $C \in \varphi$.
\end{itemize}
Clearly, $\nn$ satisfies properties~\ref{prop:Mtilde},\ref{prop:partition}, \ref{prop:linear}, \ref{prop:ftypes} and \ref{prop:ford} of Definition~\ref{def:ORD}.

Let $\nn \in \ORDW(\varphi)$.
Let $\mm$ be the substructure of $\nn$ with universe $M^\nn$.
By property~\ref{prop:Mtilde}, $\mm$ satisfies $\varphi$.
By property~\ref{prop:partition}, $\nominals^\nn = \{o_1^\nn,\ldots,o_k^\nn\}$.
By property~\ref{prop:linear}, $(o_i^\nn,o_j^\nn) \in \order^\nn$ iff $i < j$.
Thus $(\nominals^\nn,\order^\nn)$ is isomorphic to $([1,k],\leq)$.
By property~\ref{prop:ftypes}, $f_{h'}^\nn$ is a function from $A_{h'}^\mm$ to $\nominals^\nn$, for every $1 \le h' \le h$.
By property~\ref{prop:ford}, $f_{h'}^\nn$ is a $h'$-useful labeling for $\mm$, using $ \nominalsExp^\nn$
for the natural numbers $[1,k]$ and $\order^\nn$ for the order on the natural numbers
in Definition~\ref{def:useful-labeling}, for every $1 \le h' \le h$.
Because $(\nominals^\nn,\order^\nn)$ is isomorphic to $([1,k],\leq)$, the last property implies that $f_{h'}^\nn$ is a useful labeling, for every $1 \le h' \le h$.
\end{proof}

\begin{lemma}
\label{lem:definingORD}
For every formula $\varphi \in \Loiq$ there exists a formula $\useful(\varphi)$ such that $\useful(\varphi)$ defines $\ORDW(\varphi)$.
\end{lemma}
\begin{proof}
We set $\useful(\varphi) = \theta^{1} \land \theta^{2} \land \theta^{3} \land \theta^{4} \land \theta^{5a} \land \theta^{5b}$.
$\theta^{X}$ defines the property $X$ in Definition \ref{def:ORD}.
\begin{itemize}
\item
For every atomic concept $A$, let $g(A) = A \sqcap M$.
For every concept $C$, $g(C)$ is obtained by replacing its sub-concepts with their $g$ image and intersecting with $M$ (e.g., $g(C_{1}\sqcup C_{2}) = (g(C_{1})\sqcup g(C_{2})) \sqcap M$).
Let $g(\varphi)$ be obtained from $\varphi$ by replacing every inclusion $C \sqsubseteq D$ in $\varphi$ by $g(C) \sqsubseteq g(D)$.
Let $\mm$ be the substructure of $\nn$ with universe $M^\nn$ by $\mm$.
We have $\mm \models \varphi$ iff $\nn \models g(\varphi)$ (this holds because we have $C^\mm = g(C)^\nn$ for all $C \in \varphi$).

%
\item
Let $\theta^{2} = \neg M \equiv (o_1 \sqcup \cdots \sqcup o_k)$.
$\theta^{2}$ says that the universe of $N \setminus M^\nn = \nominals^\nn  = \{o_1^\nn,\ldots,o_k^\nn\}$.
\item
Let $\theta^{3}$ be the conjunction of
$o_i \sqsubseteq \left(\neg
\bigsqcup_{1 \leq  j \leq i} \exists \order. o_j\right)$
and
$o_i \sqsubseteq
\mybigsqcapDisp_{i <  j \leq k} \exists \order.o_j$  for $1\leq i \leq k$.
$\theta^{3}$ says that $o_i^\nn \in \bigcap_{i < j \leq k} \{ u \mid (u,o_j^\nn) \in \order^\nn\}$ and $o_i^\nn \not\in \bigcup_{1 \le j \leq i} \{ u \mid (u,o_j^\nn) \in \order^\nn\}$, i.e., $(o_i^\nn,o_j^\nn) \in \order^\nn$ iff $i < j$.
\item
Let $\theta^{4}$ be the conjunction of the formulas $\left(\exists f_{h'} \equiv A_{h'} \sqcap M \right)$ and $\left(\exists f_{h'}^{-}\sqsubseteq\neg M\right)$ for $1\leq h'\leq h$.
$\theta^{4}$ says that $f_{h'}^\nn$ is a function from $A^\nn \cap M^\nn = A_{h'}^\mm$ to $N \setminus M^\nn = \nominals^\nn$.
\item
Let $\theta^{5a}$ be the conjunction of
$\left(\exists f_{h'}^{-}.C\right) \sqcap \left(\exists f_{h'}^{-}. \neg C\right) \equiv \bot$ for all $1\leq h'\leq h$ and $C \in \varphi$.
$\theta^{5a}$ says that if $u,v\in A_{h'}^\mm$ point to
the same nominal (i.e., $f_{h'}^\nn(u) = f_{h'}^\nn(v)$), then they must agree on every concept $C \in \varphi$ (i.e., $u \in C^\nn$ iff $v \in C^\nn$), thus $\ovtpn_{\mm}^{\varphi}(u) = \ovtpn_{\mm}^{\varphi}(v)$.
\item
Let $\theta^{5b}$ be the conjunction, for all $1 \leq \ell \leq k$ and
\[
\begin{array}{l}
1\leq h'\leq h$, of $(\exists f_{h'}. o_\ell  \sqcap \neg B_{h'} \not\equiv \bot)
\rightarrow \\
\big(\bot \not\equiv \big(\bigsqcup_{s\in S_{h'}} \exists s. \exists f_{h'}. o_{\ell}\big) \sqcap \exists f_{h'}. \exists \order. o_{\ell} \big)
\end{array}
\]
$\theta^{5b}$ says that if $u \in A_{h'}^{\mm} \setminus B_{h'}^{\mm}$ is pointing to some nominal $o_\ell^\nn$ with $f_{h'}^\nn$, then there is a $v \in A_{h'}^{\mm}$ that has a successor pointing to the same nominal $o_\ell^\nn$ with $f_{h'}^\nn$ and that is pointing to a smaller nominal with $f_{h'}^\nn$ (i.e., pointing to some nominal in $(\exists \order. o_{\ell})^\nn$).
\end{itemize}
\end{proof}

\begin{theorem}
\label{th:exp}
Let $\Phi_{i}=\varphi_{i}\land\bigwedge RE_{i}\land\bigwedge DI_{i}\in\L$, for $i=1,2$.
There are $\LoiqNOb$ formulas $\mu$ and $\kappa$ over an extended vocabulary such that\\
(1) $\Phi_{1}$ is satisfiable iff $\mu$ is satisfiable.\\
(2) $\Phi_{1}$ implies $\Phi_{2}$ iff $\kappa$ is not satisfiable.
\end{theorem}

\begin{proof}
(1) follows from Lemmas \ref{lem:connected-implies-useful-order--},
\ref{lem:semi-def}, \ref{lem:extension-useful}, and \ref{lem:definingORD} by setting $\mu = \useful(\semi{\Phi})$.

(2): For every $h'$, let $X_{h'}$ be a fresh atomic concept and let
$\alpha_{h'}$ be the conjunction of $(B_{h'}\sqsubseteq X_{h'}) \land (A_{h'}\sqcap\neg X_{h'}\not\equiv\bot)$ and
$\bigwedge_{s\in S_{h'}}(\exists s.\neg X_{h'}\sqsubseteq\neg X_{h'})$.
For every $\mm$, $\mm\models\neg\Reacha{B_{h'}}{S_{h'}}{A_{h'}}$
iff there is $X_{h'}^{\mm}$ such that $\left\langle \mm,X_{h'}^{\mm}\right\rangle \models\alpha_{h'}$.
Hence, $\mm\models\neg\Phi_{2}$ iff there are $X_{h'}^{\mm}\subseteq M$, $1\leq h'\leq h$, such that $\left\langle \mm,X_{h'}^{\mm}:1\leq h'\leq h\right\rangle \models\neg\varphi_{2}\lor\neg\bigwedge DI_{2}\lor\bigvee_{1\leq h'\leq h}\alpha_{h'}$.
Hence, $\kappa_{\varphi} = \Phi_{1}\land \big(\neg\varphi_{2} \lor \neg \bigwedge DI_{2}\lor\bigvee_{1\leq h'\leq h}\alpha_{h'}\big)$
is satisfiable iff $\Phi_{1}\to\Phi_{2}$ is not a tautology and we
get (2).
In both (1) and (2) we use that satisfiability in $\Loiq$ is reducible to that in $\LoiqNOb$, see Appendix \ref{app:Loiq_LoiqNOb}.
\end{proof}

\subsection{NEXPTIME decision procedures \label{se:poly}}

The algorithm in Theorem~\ref{th:exp} produces, for a formula $\varphi$, a formula whose size is exponential in the size of $\varphi$.
Most of the constructions along the proof introduce only a polynomial growth, except for the nominals in Definition~\ref{def:ORD} and the formulae that use them.
We discuss here how to effectively compute an $\LoiqNOb$-formula of polynomial size in $\varphi$, which introduces the required linear ordering of exponential length without use of the nominals.
Since satisfiability in $\LoiqNOb$ is NEXPTIME-complete \cite{thesis:Tobias01}, so is satisfiability and implication in $\L$.
We sketch the idea first.

In Section~\ref{se:reduce-to-formula} a structure $\nn \in \ORDW(\varphi)$ with universe $N$ represents a model $\mm$ of $\varphi$ with universe $M$ and at the same time also contains useful labelings for $\mm$.
Here, we define structures $\nn$ which extend $\mm$ in a different though similar way.
Let $k = |\{C \mid C \in \varphi \}|$.
We introduce new concepts $P_1,\ldots,P_k$ and use them to require that $O := N \backslash M$ is of size $2^k$ and  that $\suc$ is interpreted as a successor relation on $O$.
We think of the reflexive-transitive closure of $\suc^\nn$ as $\order^\nn$ from Definition~ \ref{def:ORD}, but we will not compute $\order^\nn$ explicitly.
For every binary word $b_1 \ldots b_k$, there will be exactly one element of $O$ in $\bigcap_{i:b_i=1}P_i^\nn \cap \bigcap_{i:b_i=0} \neg P_i^\nn$.
I.e., $P_i^{\mm}$ represents elements whose corresponding binary word has $b_i=1$.
$\suc^\nn$ will be induced by the usual successor relation on binary words of length $k$:
an element $u \in O$ is the successor of $v \in O$ in $\suc^\nn$ iff 
there is $\ell$ such that (1) $u$ and $v$ agree on $P_i^\nn$, $i > \ell$, 
(2) $u \in P_\ell^\nn$ and $v \notin P_\ell^\nn$ and 
(3) $v \in P_i^\nn$ and $u \notin P_i^\nn$, $i<\ell$.  
Similarly as in Definition~ \ref{def:ORD}, the functions $f_{h'}^\nn$ need to be useful labelings, using $O$ for the numbers $[1,2^k]$ and $(\suc^\nn)^{*}$ for the linear order on natural numbers in Definition~\ref{def:useful-labeling}.
Importantly, we do not define the transitive closure $\left(\suc^\nn\right)^{*}$ explicitly.
Instead, we exploit the fact that $b_{y}\ldots b_{1}$ is less than $d_{y}\ldots d_{1}$ iff there exists an index $i$ such that $b_{y}\ldots b_{i+1}=d_{y}\ldots d_{i+1}$, $b_{i}=0$ and $d_{i}=1$.

\begin{definition}
\label{def:ORD-exp}
Let $\varphi$ be a $\Loiq$ formula over vocabulary $\tau$.
Let $k = |\{C \mid C \in \varphi \}|$.
We define an extended vocabulary $\usefulExp(\tau)$ that extends $\tau$ with new atomic concepts $M, P_1,\ldots,P_k$, a new nominal $o_\st$, and new atomic functional roles $\suc, f_1,\ldots, f_h$.

Let $\nn$ be a $\usefulExp(\tau)$-structure with universe $N$.
We denote the substructure of $\nn$ with universe $M^\nn$ by $\mm$.
We denote the set $N \setminus M^\nn$ by $\nominalsExp^\nn$.
We denote by $\eval: \nominalsExp^\nn \rightarrow [1,2^k]$ the function that maps an element $u \in \nominalsExp^\nn$ to $\eval(u) = 1 + \sum_{i:u \in P_i^\nn} 2^{i-1}$.
We denote by $\left(\suc^\nn\right)^{*}$ the reflexive-transitive closure of $\suc^\nn$.
The structure $\nn$ belongs to $\ORDWExp(\varphi)$ if the following conditions hold:
\begin{enumerate}
\item \label{prop:Mtilde-exp} $\mm$ satisfies $\varphi$.
\item \label{prop:partition-exp} We have $\nominalsExp^\nn = \{o_\st^\nn\} \cup \bigcup_{1 \le i \le k} P_i^\nn$.
\item \label{prop:linear-exp}
    $\eval$ is a bijective function, $\eval(o_\st^\nn) = 1$, and
    $\suc(u) = v$ iff $\eval(u) + 1 = \eval(v)$ for all $u,v \in \nominalsExp^\nn$.
\item \label{prop:ftypes-exp} $f_{h'}^\nn$ is a function from $A_\ell^\mm$ to $\nominals^\nn$, for every $1 \le {h'} \le h$.
\item \label{prop:ford-exp} $f_{h'}^\nn$ is a $h'$-useful labeling for $\mm$, using $ \nominalsExp^\nn$ for the natural numbers $[1,2^k]$ and $\left(\suc^\nn\right)^{*}$ for the order on the natural numbers in Definition~\ref{def:useful-labeling}, for every $1 \le h' \le h$.
\end{enumerate}
\end{definition}

\begin{lemma}
\label{lem:extension-useful-exp}
$\ORDWExp(\varphi)$ is non-empty iff there is a model $\mm$ of $\varphi$ with $h'$-useful labelings for $\mm$ for every $1\leq h'\leq h$.
\end{lemma}

\begin{lemma}
\label{lem:definingORD-exp}
For every formula $\varphi \in \Loiq$ there exists a formula $\usefulExp(\varphi)$, of size polynomial in $k$, such that $\usefulExp(\varphi)$ defines $\ORDWExp(\varphi)$.
\end{lemma}

\begin{proof}
We set $\useful(\varphi) = \theta^{1} \land \theta^{2} \land \theta^{3} \land \theta^{4} \land \theta^{5a} \land \theta^{5b}$.
$\theta^{X}$ defines the property $X$ in Definition \ref{def:ORD-exp}.
The formulae $\theta^{1}$, $\theta^{4}$ and $\theta^{5a}$ are the same as in the proof of Lemma~\ref{lem:definingORD}.
\begin{itemize}
\item
Let $\theta^{2} = \neg M \equiv o_\st \sqcup \displaystyle{\bigsqcup_{1 \le i \le k}} P_i$.
$\theta^{2}$ says that $\nominalsExp^\nn = N \setminus M^\nn = \{o_\st^\nn\}  \cup \bigcup_{1 \le i \le k} P_i^\nn$.
\item
Let $\theta^{3} = \zeta_{consec} \land \zeta_{first} \land \zeta_{last}$.
$\zeta_{consec}$ expresses that the successor relation mimics the binary words:
two words $b_{k}\ldots b_{1}$ and $d_{k}\ldots d_{1}$ are consecutive in $\suc$ iff there exists an index $i$ such that $b_{i}\ldots b_{1} = 01^{i-1}$, $d_{i}\ldots d_{1} = 10^{i-1}$, and $b_{k}\ldots b_{i+1}=d_{k}\ldots d_{i+1}$.
We introduce concepts $C_i$, for every $1\leq i\leq k$ (the concepts $C_i$ can be either added as fresh concepts or used as abbreviations; the resulting formula $\zeta_{consec}$ will be of polynomial size in both cases):
\[
\begin{array}{lll}
C_{i} & = & \neg P_{i}\sqcap\exists \suc.P_{i}\\
C_{<i} & = & {\displaystyle \mybigsqcapDisp_{j<i}\left(P_{j}\sqcap\exists \suc.\neg P_{j}\right)}\\
C_{>i} & = & {\displaystyle \mybigsqcapDisp_{i<j\leq y}\left(P_{j}\sqcap\exists \suc.P_{j}\sqcup\neg P_{j}\sqcap\exists \suc.\neg P_{j}\right)}\\
\zeta_{consec} & = & {\displaystyle (\neg M \sqcap \neg \left( P_{1} \sqcap \cdots \sqcap P_{k} \right) \sqsubseteq}\\
& & \quad \quad \quad \displaystyle{\bigsqcup_{1\leq i\leq k}} C_{<i}\sqcap C_{i}\sqcap C_{>i})
\end{array}
\]
The formulae
\[
    \begin{array}{lll}
    \zeta_{first} & = & o_\st \equiv \left(\neg P_{1} \sqcap \cdots \sqcap \neg P_{k} \right) \land\\
    & & \quad \quad \quad \exists \suc^{-}\equiv\neg M \sqcap o_\st\\
    \zeta_{last} & = & \exists \suc\equiv\neg M \sqcap \neg \left( P_{1} \sqcap \cdots \sqcap P_{k} \right)
    \end{array}
\]
specify that all elements in $\nominalsExp^\nn$ except for $(P_{1} \sqcap \cdots \sqcap P_{k})^\nn$ have a successor, that all elements except for $o_\st^\nn$ have a predecessor and that $\left(\neg P_{1}\sqcap\cdots\sqcap\neg P_{y}\right)^\nn$ contains exactly the single element $o_\st^\nn$.
The above stated facts imply that for every vector $(b_1,\ldots,b_k) \in \{0,1\}^k$ there is exactly one element of $\nominalsExp^\nn$ in
\[
(\mybigsqcapDisp_{i:b_{i}=1} P_{i} \sqcap \mybigsqcapDisp_{i:b_{i}=0} \neg P_{i})^\nn.
\]
\item
We do not define the transitive closure $\left(\suc^\nn\right)^{*}$ explicitly.
Instead, we exploit the fact that $b_{y}\ldots b_{1}$ is less than $d_{y}\ldots d_{1}$ iff there exists an index $i$ such
that $b_{y}\ldots b_{i+1}=d_{y}\ldots d_{i+1}$, $b_{i}=0$ and $d_{i}=1$.
We introduce concepts $E_{h',s,i}^\nn$, for every $1 \le h' \le h, s\in S_{h'}$ and $ \in [1,k]$, which will contain all of the elements $u \in M$ such that the types of $u$ and $s^\nn(u)$ agree on membership in $P_{i+1}^\nn,\ldots,P_{y}^\nn$, $u\notin P_{i}^\nn$ and $s^\nn(u)\in P_{i}^\nn$ (the concepts $E_{h',s,i}^\nn$ can be either added as fresh concepts or used as abbreviations; the resulting formula $\theta^{5b}$ will be of polynomial size in both cases).
\begin{eqnarray*}
E_{h',s,i} & \equiv & \exists f_{h'}.\neg P_{i} \sqcap\exists s.\exists f_{h'}.P_{i} \sqcap \\
& &\mybigsqcapDisp_{i+1\leq j\leq y} \Big( \exists f_{h'}.P_{j}\sqcap\exists s.\exists f_{h'}.P_{j} \sqcup\\
& & \quad \quad \exists f_{h'}.\neg P_{j}\sqcap\exists s.\exists f_{h'}.\neg P_{j}\Big)
\end{eqnarray*}
The formula $\theta^{5b}$ states that for every element $u\in A_{h'}^\mm$ with $u \not\in B_{h'}^{\mm}$ there is an element $v$ with the same value (i.e., $v$ points to the same element as $u$ via $f_{h'}$) such that $v$ is again in $A_{h'}^{\mm}$ and $v$ has a previous element which is smaller than $v$.
$\theta^{5b}$ is the conjunction of
\[
\begin{array}{ll}
\exists f_{h'}^{-}. \neg B_{h'} \sqsubseteq \exists f_{h'}^{-}. \left(\displaystyle{\bigsqcup_{s\in S_{h'}, i \in [1,k]}} \exists s^{-}.E_{h',s,i}\right)
\end{array}
\]
for every $1\leq h'\leq h$.
\end{itemize}
\end{proof}

\begin{theorem}
\label{th:poly} Let $\Phi_{i}=\varphi_{i}\land\bigwedge RE_{i}\land\bigwedge DI_{i}\in\L$
for $i=1,2$. 
There are polynomial-time computable $\LoiqNOb$ formulas
$\eta$ and $\rho$ over an extended vocabulary such that
\begin{enumerate}
\item $\Phi_{1}$ is satisfiable iff $\eta$ is satisfiable.
\item $\Phi_{1}$ implies $\Phi_{2}$ iff $\rho$ is not satisfiable.
\item Satisfiability and implication in $\L$ is NEXPTIME-complete.
\end{enumerate}

\end{theorem}

(1) follows from Lemmas \ref{lem:conn-equiv-semi-conn-plus-useful-ord--},
\ref{lem:semi-def}, \ref{lem:extension-useful-exp} and \ref{lem:definingORD-exp} by setting $\eta = \usefulExp(\semi{\Phi})$.

(2) follows from (1) similarly to Theorem \ref{th:exp}.

(3): We use here the reduction from $\Loiq$ to $\LoiqNOb$ in  Appendix~\ref{app:Loiq_LoiqNOb}.
Satisfiability in $\LoiqNOb$ is NEXPTIME-complete \cite{thesis:Tobias01}.
Since $\L$ contains $\LoiqNOb$, and at the same time, satisfiability
and implication of $\L$-formulae are polynomial-time reducible to
$\LoiqNOb$ satisfiability, (3) holds.

\subsection{\label{app:useful-proof} \label{app:proofs} Applying $\rhd$ leads to Standard Models}

In the proof of Lemma \ref{lem:connected-implies-useful-order--} we have shown how to turn non-standard models
into standard models by repeated aplications of $\rhd$ and based on Lemma \ref{lem:rhd-connectedness--}, which we prove here.

\begin{proof}[Proof of Lemma~\ref{lem:rhd-connectedness--}]
We assume $\val(D_{h'}^{\mm}) > 0$.
Let $X$ be a base for $D_{h'}^{\mm}$ such that $\val(X) = \val(D_{h'}^{\mm})$.
Because of $\val(D_{h'}^{\mm}) > 0$, there is a base element $a_{1} \in X$ with $a_{1} \not\in B_{h'}^\mm$.
By the $h'$-usefulness of $f_{h'}$, there are
$a_0,w\in A_{h'}^{\mm}$ such that $f_{h'}(a_{0}) = f_{h'}(a_1)$, $f_{h'}(w) < f_{h'}(a_0)$ and the graph $D_{h'}^{\mm}$ has an edge $(w,a_0)$.

By Lemma~\ref{lem:basis-circle} $a_{1}$ belongs to some cycle $C$ in $D_{h'}^{\mm}$.
We denote the successor of $a_1$ in $C$ by $b_{1}$.
Let $r \in \NF$ be the functional role with $(a_1,b_1) \in r^\mm$ and $r \in S_{h'}$ (recall that the set $S_{h'}$ belongs to the reachability-assertion $\Reacha{B_{h'}}{S_{h'}}{A_{h'}}$).
We denote the remaining path from $b_1$ to $a_1$ in $C$ by $\pi_{b_1,a_1}$.
We note that $\pi_{b_1,a_1}$ does not contain $a_0$;
otherwise $X' = X \setminus \{a_1\} \cup \{w\}$ would be a basis for $D_{h'}^{\mm}$ with $\val(D_{h'}^{\mm}) = \val(X) > \val(X')$, contradiction.
We denote by $b_0$ the $r^\mm$-successor of $a_0$, if it exists (i.e., if there is an edge $(a_0,b_0) \in r^\mm$).

We set $\mathfrak{t}=(a_{0},a_{1},r)$.
For all $1\leq\ell\leq h$ we define the shorthand $D_{\ell}^{\rhd}=D_{\ell}^{\op{\mm}{\mathfrak{t}}},\,\, A_{\ell}^{\rhd}=A_{\ell}^{\op{\mm}{\mathfrak{t}}},\,\, B_{\ell}^{\rhd}=B_{\ell}^{\op{\mm}{\mathfrak{t}}}$.

\begin{enumerate}
\item By the compatibility of $RE$ and $DI$, $D_\ell^{\mm}=D_\ell^{\rhd}$ for all $\ell \not= h'$.
\item By Lemma \ref{lem:op-on-concepts}, $\ovtpn_\mm(u)=\ovtpn_{\op{\mm}{\mathfrak{t}}}^{\varphi}(u)$ for all $u\in M$.
    Since the outgoing edges of every vertex $u \notin \{a_0,a_1\}$ do not change by applying $\rhd$,
    $\ovtpn_\mm(u)=\ovtpn_{\op{\mm}{\mathfrak{t}}}^{\varphi}(u)$.
    For $a_0,a_1$, the only change is in $s$, but $a_0$ and $a_1$ have outgoing edges corresponding to $r$ both in $\mm$ and in $\op{\mm}{\mathfrak{t}}$.
    So, $\ovtpn_\mm(u)=\ovtpn_{\op{\mm}{\mathfrak{t}}}^{\varphi}(u)$ for $u\in\{a_0,a_1\}$.

\item By Lemma \ref{lem:op-on-concepts}, $\mm$ and $\op{\mm}{\mathfrak{t}}$ agree on $\varphi$.
\item We show that $X' = X \setminus \{a_1\} \cup \{w\}$ is a base for $D_{h'}^{\rhd}$.
    We have $\val(X) > \val(X')$ by $f_{h'}(a_{0}) = f_{h'}(a_{1})$ and $f_{h'}(w) < f_{h'}(a_{0})$.
    This is sufficient to establish $\val(D_{h'}^{\mm}) = \val(X) > \val(X') \ge  \val(D_{h'}^{\rhd})$.

    We consider some node $v \in D_{h'}^{\mm}$.
    Because $X$ is a basis, $v$ is reachable from some $u \in X$ by some path $\pi$.
    We introduce $Z = \{a_{0},b_{0},a_{1},b_1\}$ as a shorthand.
    We proceed by a case distinction.

    Case 1:
        $\pi$ does not contain any node from $Z$ (in particular $u \neq a_1$).
        Then, $\pi$ also witnesses that $v$ is reachable from $u \in X'$ by $\pi$ in $D_{h'}^{\rhd}$.

    Case 2:
        $\pi$ contains a node from $Z$.
        Then there is a decomposition of $\pi$ into two paths $\pi_1$ and $\pi_2$, i.e.,
        $\pi = \pi_1\pi_2$, such that $\pi_2$ starts with a node $z \in Z$ but otherwise does not visit $Z$.
        We construct a path $\pi_0$ from $w$ to $z$ using a suitable combination of the edges $(w,a_0),(a_0,b_1),(a_1,b_0)$ and the path $\pi_{b_1,a_1}$.
        Then the composition $\pi' = \pi_0\pi_2$ witnesses that $v$ is reachable from $w \in X'$ by $\pi'$ in $D_{h'}^{\rhd}$.

    For use in 6), we point out that the special case of $v = w$ establishes that $w$ is
    either reachable from a node in $X \setminus \{a_1\}$ (case 1)
    or belongs to a cycle in $D_{h'}^{\rhd}$ (case 2).

\item Follows directly from 1 and 2.

\item Follows directly from 1 for every $1\leq \ell \leq h$ with $h'\not=\ell$.
In 4), we have established that $X' = X \setminus \{a_1\} \cup \{w\}$ is a basis for $D_{h'}^{\rhd}$ and that $w$ is either reachable from a node in $X \setminus \{a_1\}$ or belongs to a cycle.
It remains to show that every $u \in X \setminus \{a_1\}$ either belongs to $B_{h'}^\rhd$ or to a cycle of $D_{h'}^{\rhd}$.
We fix some $u \in X \setminus \{a_1\}$.
By Lemma~\ref{lem:basis-circle} $u$ either belongs to $B_{h'}^\mm$ or to a cycle of $D_{h'}^{\mm}$.
If $u \in B_{h'}^\mm$, then $u \in B_{h'}^\rhd$ by 2).
Otherwise, $u$ belongs to a cycle $C$ in $D_{h'}^{\mm}$.
$C$ cannot contain $a_1$; otherwise $X \setminus \{a_1\}$ would be a base with $\val(X \setminus \{a_1\}) < \val(X) = \val(D_{h'}^{\mm})$.
We obtain the cycle $C'$ in $D_{h'}^{\rhd}$ by replacing every edge $(a_0,b_0)$ in $C$ with the path $(a_0,b_1),\pi_{b_1,a_1},(a_1,b_0)$.

\item Follows directly from 1 and 2 for every $1\leq \ell \leq h$ with $h'\not=\ell$.
Because $f_{h'}$ is a $h'$-useful labeling for $D_{h'}^{\mm}$ we have that (1) $f_{h'}(u) = f_{h'}(v)$ implies $\ovtpn_{\mm}^{\varphi}(u) = \ovtpn_{\mm}^{\varphi}(v)$ for all $u,v \in A_{h'}^{\mm}$ and (2) for every element $u\in A_{h'}^{\mm}$, either $u \in B_{h'}^{\mm}$, or there exist elements $v,w\in A_{h'}^{\mm}$ such that $f_{h'}(u) = f_{h'}(v)$, $f_{h'}(w) < f_{h'}(v)$ and the graph $D_{h'}^{\mm}$ has an edge $(w,v)$.
We have that $\ovtpn_\mm(u) = \ovtpn_{\op{\mm}{\mathfrak{t}}}^{\varphi}(u)$ for all $u \in M$ by 2).
Thus, $A_{h'}^{\mm} = A_{h'}^{\rhd}$.
Further, $f_{h'}$ has the same values on $A_{h'}^{\mm}$ and $A_{h'}^{\rhd}$.
Because the operation $\rhd$ changed only the edges $(a_0,b_0)$ and $(a_1,b_1)$ these facts almost show that
$f_{h'}$ is a $h'$-useful labeling for $D_{h'}^{\rhd}$.
It remains to argue that for every element $u$ with $f_{h'}(u) = f_{h'}(a_0)$ ($= f_{h'}(a_1))$  there exist elements $v,w \in A_{h'}^{\rhd}$ such that $f_{h'}(u) = f_{h'}(v)$, $f_{h'}(w) < f_{h'}(v)$ and the graph $D_{h'}^{\rhd}$ has an edge $(w,v)$.
This fact is witnessed by the edge $(w,a_1)$.
\end{enumerate}
\end{proof}

\bibliographystyle{plain}

\appendix 

\section{\label{app:LoiqFO} \texorpdfstring{$\Loiq$}{\LoiqTEXT} and first order logic}

$\Loiq$ has a fairly standard reduction to the two-variable fragment
of first order logic with counting $\ctwo$ (see e.g. \cite{DBLP:journals/ai/Borgida96})
\begin{definition}
Let $tr:\Loiq\to\ctwo$ be given as follows:
\end{definition}
\[
\begin{array}{lll}
tr_{z}(C) & = & C(z)\qquad\qquad~C~\mbox{ is an atomic concept}\\
tr_{z,\bar{z}}(r) & = & r(z,\bar{z})\qquad\qquad r~\mbox{ is an atomic role}\\
tr_{z}(C\sqcap D) & = & tr_{z}(C)\land tr_{z}(D)\\
tr_{z}(C\sqcap D) & = & tr_{z}(C)\lor tr_{z}(D)\\
tr_{z}(\neg C) & = & \neg tr_{z}(C)\\
tr_{z,\bar{z}}(r^{-}) & = & tr_{\bar{z},z}(r)\\
tr_{z}(\exists r.C) & = & \exists y.tr_{z,\bar{z}}(r)\land tr_{\bar{z}}(C)\\
tr_{z}(\exists^{\leq n}r.C) & = & \exists^{\leq n}y.tr_{z,\bar{z}}(r)\land tr_{\bar{z}}(C)\\
tr_{z}(\exists^{\geq n}r.C) & = & \exists^{\geq n}y.tr_{z,\bar{z}}(r)\land tr_{\bar{z}}(C)\\
\\
tr(C\sqsubseteq D) & = & \forall x.tr_{x}(C)\rightarrow tr_{x}(D)\\
tr(\varphi\land\psi) & = & tr(\varphi)\land tr(\psi)\\
tr(\neg\varphi) & = & \neg tr(\varphi)
\end{array}
\]

\begin{lemma}
For every $\varphi\in\Loiq$, $\varphi$ and $tr(\varphi)$ agree
on the truth value of all $\tau$-structures.
\end{lemma}

\section{\label{app:Loiq_LoiqNOb} From \texorpdfstring{$\Loiq$}{\LoiqTEXT}
to \texorpdfstring{$\LoiqNOb$}{\LoiqNObTEXT}}

Here we show the reduction from satisfiability in $\Loiq$ to satisfiability
in $\LoiqNOb$.
\begin{lemma}
Let $\tau$ be a vocabulary and $\varphi\in\Loiq(\tau)$. There exist
a vocabulary $\sigma\supseteq\tau$ and $\psi\in\LoiqNOb(\sigma)$
such that $\varphi$ is satisfiable iff $\psi$ is satisfiable, and
the size of $\psi$ is linear in the size of $\varphi$. More precisely:
\begin{enumerate}
\item If $\mm$ is a $\tau$-structure satisfying $\varphi$, then there
exists an extension $\nn$ of $\mm$ such that $\nn\models\psi$.
$\nn$ has the same universe as $\mm$ and agrees with $\mm$ on the
interpretation of the symbols in $\tau$.
\item If $\nn$ is a $\sigma$-structure satisfying $\psi$, then the substructure
of $\nn$ which corresponds to $\tau$ satisfies $\varphi$.
\end{enumerate}
\end{lemma}
\begin{proof}
We prove the claim by induction on the construction of formulae in
$\Loiq$. The claim we prove is slightly augmented as follows:
\begin{itemize}
\item We assume without loss of generality that $\varphi$ is given in negation
normal form (NNF).
\item $\psi$ will not contain any negations.
\end{itemize}

We may assume without loss of generality that if $\varphi$ is satisfiable,
then it is satisfiable by a structure of size strictly larger than
$1$.

\end{proof}
\begin{description}
\item [{Base}] If $\varphi=C\sqsubseteq D$, then $C\sqsubseteq D\in\LoiqNOb$
and $\sigma=\tau$.
\item [{Closure}] Let $\varphi_{1},\varphi_{2}\in\Loiq(\tau)$ in NNF,
$\sigma_{1},\sigma_{2}\supseteq\tau$ be vocabularies, $\psi_{1}\in\mathcal{\LoiqNOb}(\sigma_{1})$
and $\psi_{2}\in\mathcal{\LoiqNOb}(\sigma_{2})$ as guaranteed. Without
loss of generality, $(\sigma_{1}\backslash\tau)\cap(\sigma_{2}\backslash\tau)=\emptyset$.

\begin{enumerate}
\item $\varphi=\varphi_{1}\land\varphi_{2}$: Let $\psi=\psi_{1}\land\psi_{2}$
and $\sigma=\sigma_{1}\cup\sigma_{2}$.
\item $\varphi=\neg\varphi_{1}$: By the assumption that that $\varphi$
is in NNF, $\varphi_{1}$ is of the form $(C\sqsubseteq D)$. Let
$o$ be a fresh nominal which does not occur in $\sigma_{1}$. Let
$\sigma=\sigma_{1}\cup\{o\}$. Let $\psi=(o\sqsubseteq C)\land(D\sqsubseteq\neg o)$.
\item $\varphi=\varphi_{1}\lor\varphi_{2}$. Let $r$ be a fresh role and
$o_{1},o_{2},o_{X},o_{Y}$ be fresh nominals.
\begin{eqnarray*}
\psi_{prep} & = & (o_{X}\sqsubseteq\neg o_{Y})\land(o_{1}\sqsubseteq\neg o_{2})\land(o_{X}\sqcup o_{Y}\equiv o_{1}\sqcup o_{2})\land\\
 &  & (\exists r.o_{X}\equiv\top)\land(\exists r.o_{Y}\equiv\bot)
\end{eqnarray*}
For a structure $\mm$ with universe $M$, $\mm\models\psi_{prep}$
iff

\begin{enumerate}
\item $o_{X}^{\mm}\not=o_{Y}^{\mm}$
\item $o_{1}^{\mm}\not=o_{2}^{\mm}$
\item $o_{1}^{\mm}=o_{X}^{\mm}$ and $o_{2}^{\mm}=o_{Y}^{\mm}$, or $o_{1}^{\mm}=o_{Y}^{\mm}$
and $o_{2}^{\mm}=o_{X}^{\mm}$.
\item $(\exists r.o_{1})^{\mm}=M$ and $(\exists r.o_{2})^{\mm}=\bot$,
or

$(\exists r.o_{1})^{\mm}=\bot$ and $(\exists r.o_{2})^{\mm}=M$.

\end{enumerate}

For $i\in\{1,2\}$, let $\theta_{i}$ be the formula obtained from
$\varphi_{i}$ by replacing every atomic formula $C\sqsubseteq D$
with $C\sqcap\exists r.o_{i}\sqsubseteq D\sqcap\exists r.o_{i}$.
Let $\sigma=\sigma_{1}\cup\sigma_{2}\cup\{r,o_{1},o_{2},o_{X},o_{Y})$.
Let $\psi=\psi_{prep}\land\theta_{1}\land\theta_{2}$. The desired
property follows directly from the claim:
\begin{claim}
Let $\nn$ be a $\sigma$-structure such that $\nn\models\psi_{prep}$,
and let $\mm_{\nn}$ be the substructure of $\nn$ which corresponds
to $\tau$.
\begin{enumerate}
\item If $(\exists r.o_{1})^\nn=M$, then $\nn\models\psi$ iff $\mm_{\nn}\models\varphi_{1}$.
\item If $(\exists r.o_{1})^\nn=\emptyset$, then $\nn\models\psi$ iff
$\mm_{\nn}\models\varphi_{2}$. \end{enumerate}
\begin{proof}
~\end{proof}
\begin{enumerate}
\item Let $\nn$ be a $\sigma$-structure such that $(\exists r.o_{1})^\nn=M$.
For every atomic formula $C\sqsubseteq D$ in $\varphi_{1}$, $(C\sqcap\exists r.o_{1})^\nn=C^\nn\cap M=C^\nn$
and $(D\sqcap\exists r.o_{1})^\nn=D^\nn\cap M=D^\nn$. Hence,
$\mm_{\nn}\models C\sqsubseteq D$ iff $\nn\models C\sqcap\exists r.o_{1}\sqsubseteq D\sqcap\exists r.o_{1}$.
By construction of $\theta_{1}$, $\mm\models\varphi_{1}$ iff $\nn\models\theta_{1}$.

For every atomic formula $C\sqsubseteq D$ in $\varphi_{2}$, $(C\sqcap\exists r.o_{2})^\nn=C^\nn\cap\emptyset=\emptyset$
and $(D\sqcap\exists r.o_{2})^\nn=D^\nn\cap\emptyset=\emptyset$.
Hence, $\nn\models C\sqcap\exists r.o_{2}\sqsubseteq D\sqcap\exists r.o_{2}$.
Since $\theta_{2}$ is a negation free Boolean combination of atomic
formulae, $\nn\models\theta_{2}$.

\item This case is symmetric to the previous case. Let $\nn$ be a $\sigma$-structure
such that $(\exists r.o_{1})^\nn=\emptyset$. Then $(\exists r.o_{2})^\nn=M$.
For every atomic formula $C\sqsubseteq D$ in $\varphi_{2}$, $(C\sqcap\exists r.o_{2})^\nn=C^\nn\cap M=C^\nn$
and $(D\sqcap\exists r.o_{2})^\nn=D^\nn\cap M=D^\nn$. Hence,
$\mm\models C\sqsubseteq D$ iff $\nn\models C\sqcap\exists r.o_{2}\sqsubseteq D\sqcap\exists r.o_{2}$.
By construction of $\theta_{2}$, $\mm_{\nn}\models\varphi_{2}$ iff
$\nn\models\theta_{2}$.

For every atomic formula $C\sqsubseteq D$ in $\varphi_{1}$, $(C\sqcap\exists r.o_{1})^\nn=C^\nn\cap\emptyset=\emptyset$
and $(D\sqcap\exists r.o_{1})^\nn=D^\nn\cap\emptyset=\emptyset$.
Hence, $\nn\models C\sqcap\exists r.o_{1}\sqsubseteq D\sqcap\exists r.o_{1}$.
Since $\theta_{1}$ is a negation free Boolean combination of atomic
formulae, $\nn\models\theta_{1}$.

\end{enumerate}
\end{claim}
\end{enumerate}
\end{description}

\section{\label{app:proofs-lem:op-on-concepts} Proof of Lemma \ref{lem:op-on-concepts}}

The proof of (1) proceeds by induction on the construction of the concepts, showing that $C^{\mm} = C^{\op{\mm}{\mathfrak{t}}}$ for all $C\in\varphi$.
For ease of notation we write $\mm_{1}=\mm$ and $\mm_{2} = \op{\mm}{\mathfrak{t}}$.
\begin{enumerate}
\item If $A\in\NC$, then $A^{\mm_{1}} = A^{\mm_{2}}$ since none of the atomic concepts change between $\mm_{1}$ and $\mm_{2}$.
\item If $o\in\NI$, then similarly, there is no change.
\item If $C_{1}$ and $C_{2}$ are concepts satisfying the induction hypothesis, then $C_{1} \sqcap C_{2}$, $C_{1} \sqcup C_{2}$ and $\neg C_{1}$ also satisfy the claim, e.g.,
    $(C_{1} \sqcap C_{2})^{\mm_{1}} = C_{1}^{\mm_{1}} \cap C_{2}^{\mm_{1}} =
     C_{1}^{\mm_{2}} \cap C_{2}^{\mm_{2}} = (C_{1} \sqcap C_{2})^{\mm_{2}}$.
\item For a role $s$, a concept $C$ and a non-negative integer $n$, we consider the concepts $\exists s.C$, $\exists^{\leq n} s.C$, $\exists s^{-}.C$ and $\exists^{\leq n} s^{-}.C$:

    \begin{enumerate}
    \item If $s\not=r$, then $\left(\exists s.C^{\mm_{1}}\right) = \left(\exists s.C\right)^{\mm_{2}}$ since $s^{\mm_{1}}=s^{\mm_{2}}$ and by induction $C^{\mm_{1}} = C^{\mm_{2}}$.
        Similarly, this holds for $\exists^{\leq n} s.C$, $\exists s^{-}.C$ and $\exists^{\leq n} s^{-}.C$.
    \item If $s=r$:
    \begin{itemize}
    \item $\exists r.C$ and $\exists^{\leq n} s.C$:
        We fix some $i \in \{0,1\}$.
        We have $a_i \in \left(\exists r.C\right)^{\mm_{1}}$ iff
        $a_{i-1} \in \left(\exists r.C\right)^{\mm_{1}}$
        (by $\ovtpn_{\mm_{1}}^{\varphi}(a_{0}) = \ovtpn_{\mm_{1}}^{\varphi}(a_{1})$)
        iff there is a $b$ such that $(a_{i-1},b) \in r^{\mm_{1}}$ and $b \in C^{\mm_{1}}$
        iff there is a $b$ such that $(a_{i-1},b) \in r^{\mm_{1}}$ and $b \in C^{\mm_{2}}$ (by induction assumption) iff
        iff there is a $b$ such that $(a_i,b) \in r^{\mm_{2}}$ and $b \in C^{\mm_{2}}$ (by the definition of the operation $\rhd$) iff
        $a_i \in \left(\exists r.C\right)^{\mm_{2}}$.
        Since the only difference between $\mm_{1}$ and $\mm_{2}$ are the values of $r^{\mm_{i}}$ on $a_{0}$ and $a_{1}$, we have $\left(\exists r.C\right)^{\mm_{1}} = \left(\exists r.C\right)^{\mm_{2}}$
        and $\left(\exists^{\leq n} r.C\right)^{\mm_{1}} = \left(\exists^{\leq n} r.C\right)^{\mm_{2}}$.
    \item $\exists r^{-}.C$ and $\exists^{\leq n} r^{-}.C$:
        For every $u\in M$, $i=1,2$, we define
        \[
            S_{i}(u)=\left\{ v\mid(u,v)\in\left(r^{-}\right)^{\mm_{i}}\mbox{ and }v\in C^{\mm_{i}}\right\}
        \]
        We fix some $u\in M$.
        We have $v\in S_{1}(u)$ iff $v\in S_{2}(u)$ for every $v \notin \{a_{0},a_{1}\}$, using that $(u,v) \in  \left(r^{-}\right)^{\mm_{1}}$
        iff $(u,v) \in \left(r^{-}\right)^{\mm_{2}}$ and that by induction $C^{\mm_{1}}=C^{\mm_{2}}$.
        We now consider $v \in \{a_{0},a_{1}\}$:
        We have $a_{i} \in S_{1}(u)$
        iff $(a_{i},u) \in r^{\mm_{1}}$ and $a_{i}\in C^{\mm_{1}}$
        iff $(a_{i},u) \in r^{\mm_{1}}$ and $a_{i-1}\in C^{\mm_{1}}$ (because $\ovtpn_{\mm_1}^{\varphi}(a_{0}) = \ovtpn_{\mm_1}^{\varphi}(a_{1})$)
        iff $(a_{i-1},u) \in r^{\mm_{2}}$ and $a_{i-1}\in C^{\mm_{1}}$ (by the definition of the operation $\rhd$)
        iff $(a_{i-1},u)\in r^{\mm_{2}}$ and $a_{i-1}\in C^{\mm_{2}}$ (by induction assumption)
        iff $a_{i-1}\in S_{2}(u)$.
        So, $|S_{1}(u)| = |S_{2}(u)|$.
        Therefore, $u \in (\exists r^{-}.C)^{\mm_{1}}$ iff $u \in (\exists r^{-}.C)^{\mm_{2}}$ and $u \in (\exists^{\leq n} r^{-}.C)^{\mm_{1}}$ iff $u \in (\exists^{\leq n} r^{-}.C)^{\mm_{2}}$, $i=1,2$.
    \end{itemize}
\end{enumerate}
\end{enumerate}
We turn to the two conclusions:
\begin{enumerate} [(2)]
    \item We get directly that for every $u\in M$, $\ovtpn_{\mm}^{\varphi}(u) = \ovtpn_{\op{\mm}{\mathfrak{t}}}^{\varphi}(u)$.
\end{enumerate}
\begin{enumerate} [(3)]
    \item Follows using Lemma~\ref{lem:types}.
\end{enumerate}

\end{document}